\documentclass[twocolumn,showpacs,preprintnumbers,amsmath,amssymb,amsfonts,final,a4paper,aps,citeautoscript,footinbib,pra,superscriptaddress]{revtex4-2}
\usepackage{amsthm}

\newtheorem{lem}{\uline{Lemma}}
\newtheorem{prop}{\uline{Proposition}}
\newtheorem{cor}{\uline{Corollary}}
\usepackage{braket}
\usepackage{physics}
\usepackage[pdftex]{graphicx,color}
\usepackage{amssymb,amsmath,amsfonts}
\graphicspath{{./figures/}}
\usepackage[latin1]{inputenc}
\usepackage{mathrsfs}
\usepackage{mathtools}
\usepackage{epstopdf}
\usepackage{times}
\usepackage{float}
\usepackage[normalem]{ulem}
\usepackage[sort&compress]{natbib}
\usepackage[colorlinks,bookmarks=true,citecolor=blue,linkcolor=blue,urlcolor=blue, breaklinks=true]{hyperref}

\newcommand{\eq}[1]{\begin{equation}{#1}\end{equation}}

\newcommand{\sun}[1]{\mathfrak{su}([1])}

\newcommand{\bol}[1]{\mathbf{#1}}
\newcommand{\bolm}{\boldsymbol{m}}
\newcommand{\bolP}{\boldsymbol{\mathcal{P}}}
\newcommand{\be}{\mathrm{e}}

\newtheorem*{propwonum}{Proposition}
\begin{document}
\title{Computational Characterization of Symmetry-Protected\\
Topological Phases in Open Quantum Systems}

\author{Riku Masui}
\affiliation{Division of Physics and Astronomy, Graduate School of Science, Kyoto University, Kyoto 606-8502, Japan}
\affiliation{Yukawa Institute for Theoretical Physics, Kyoto University, Kyoto 606-8502, Japan}
\author{Keisuke Totsuka}
\affiliation{Yukawa Institute for Theoretical Physics, Kyoto University, Kyoto 606-8502, Japan}

\date{\today}
\begin{abstract}
It is a challenging problem to correctly characterize the symmetry-protected topological (SPT) phases in open quantum systems.  
As the measurement-based quantum computation (MBQC) utilizes non-trivial edge states of the SPT phases as the logical qubit, 
its computational power is closely tied to the non-trivial topological nature of the phases.  
In this paper, we propose to use the gate fidelity which is a measure of the computational power of the MBQC to identify 
the SPT phases in mixed-state settings.  
Specifically, we investigate the robustness of the Haldane phase by considering the MBQC on the Affleck-Kennedy-Lieb-Tasaki state 
subject to different types of noises.  
To illustrate how our criterion works, 
we analytically and numerically calculated the gate fidelity to find that its behavior depends crucially on whether the noises satisfy 
 a certain symmetry condition with respect to the on-site $\mathbb{Z}_2 \times \mathbb{Z}_2$ symmetry.  
In particular, the fidelity for the identity gate, which is given by the sum of the non-local string order parameters, plays an important role.  
Furthermore, we demonstrate that a stronger symmetry conditions are required to be able to perform 
other (e.g.,  the $Z$-rotation gate) gates with high fidelity.  
By examining which unitary gates can be implemented with the MBQC on the decohered states, 
we can gain a useful insight into the richer structure of noisy SPT states that cannot be captured solely by the string order parameters.  
\end{abstract}

\maketitle

\section{Introduction}
\label{sec:intro}
Studying the character of many different quantum phases is mainly based on classification of gapped Hamiltonians, which is equivalent to the classification of gapped ground states. Some quantum phases exist without spontaneous breaking of global symmetries, and we cannot characterize them with local order parameters such as the magnetization. These relatively new phases (often said as ``beyond Landau paradigm'') which we call ``topological'' 
include quantum spin liquids, the quantum Hall states, the fractonic states, and so on \cite{Wen-RMP-17,Pretko-C-Y-review-20}.   
One useful way to distinguish these phases from other trivial ones is to use quantum entanglement \cite{Chen-G-W-10}.    
In this respect, ``trivial phases'' are those which can be deformed into product states by local unitary transformations (or finite-depth quantum circuits). 
In contrast, a state that cannot be transformed into a product state by such transformations without closing the gap is referred to as long-range entangled. 
In this sense, it is known \cite{Chen-G-W-11,Schuch-PG-C-11} that long-range entangled states with ``genuine'' topological order 
are forbidden in one-dimensional (bosonic) systems.  
However, if we consider only local unitary transformations respecting a certain symmetry, 
some sets of states become distinct from the set of product states.  
This is how the concept of the symmetry-protected topological (SPT) phases arises \cite{Gu-W-09}.  
One of the best known examples of the SPT phases in one dimension would be the Haldane phase \cite{Haldane_conjecture,Haldane_conjecture-2}, 
which is protected by one of (i) spatial inversion, (ii) time-reversal, and (iii) $\mathbb{Z}_2 \times \mathbb{Z}_2$ symmetries. For example, the ground state of the spin-1 antiferromagnetic Heisenberg model and the Affleck-Kennedy-Lieb-Tasaki (AKLT) state \cite{AKLT} belong to this phase. 
The Haldane phase has no local order parameter, and instead, is characterized by, e.g., 
the non-vanishing string order parameter \cite{string_order,Garcia-W-S-V-C-08} which is associated with the $\mathbb{Z}_2 \times \mathbb{Z}_2$ symmetry, 
the even-fold degeneracy of the entanglement spectrum \cite{Pollmann-T-B-O-10}, 
and the emergent fractionalized edge spins \cite{Hagiwara-K-A-H-R-90,Kennedy-90,Glarum-G-L-K-M-91} which are essentially different from those constituting the bulk.    

In one dimension, gapped ground states are well approximated by the matrix product states (MPS) \cite{Perez-Garcia-V-W-C-07} 
which automatically satisfy the area law of the entanglement entropy \cite{Hastings-07}, and the aforementioned non trivial properties 
of the Haldane phase can be altogether understood using the MPS representation \cite{Pollmann-T-B-O-10,Schuch-PG-C-11,Chen-G-W-11,Pollmann-T-12}.  
The most important message from this MPS-based approach would be 
that the symmetry group acts on the fractionalized edge state as a projective representation \cite{Garcia-W-S-V-C-08}.  
Therefore, when the group $G$ of the on-site symmetry is given, the classification of one-dimensional $G$-symmetric SPT phases 
amounts to enumerating all the possible 
projective representations of $G$ \cite{Chen-G-W-11,Schuch-PG-C-11,Chen-G-L-W-13}.  

On the other hand, the measurement-based quantum computation (MBQC) \cite{Raussendorf-B-01,Brennen-M-08} is a computational model known to be mathematically equivalent to the quantum circuit model. In this formulation, the quantum teleportation in which the state sent  
by one person is the same as that received by the other is applied. The sender teleports some qubit to the receiver by the specific two-qubit measurement based on the entangled state which the sender and receiver share. The MBQC utilizes the fact that the unitarily transformed state can be sent by appropriate measurements or by sharing appropriate entangled state between the sender and the receiver \cite{Nielsen-03,Gottesman-C-99}.  
Therefore, in the MBQC, the computation is literally implemented by the measurement on what is called the resource states. 
It is known \cite{Doherty-B-09,Else-S-B-D-12,Miyake-10,Stephen-W-P-W-R-17,Raussendorf-W-P-W-S-17,Wang-S-R-17} 
that states in non-trivial SPT phases can be used as resource states of MBQC (the so-called computational phases of matter), and 
in simplest cases, the corresponding qubits are encoded in the physical edge states of the SPT phases \cite{Miyake-10,Stephen-W-P-W-R-17}.  
Furthermore, the computational power is retained throughout a given SPT phase and the feasibility of MBQC is thought of as associated  
not with a particular state but with the phase itself \cite{Else-S-B-D-12,Raussendorf-W-P-W-S-17}.  

Then, it would be natural from both topological and quantum-computational points of view to ask how robust SPT phases are against 
the coupling to the environment. 
While, in the case of pure states, SPT phases, have been fairly well understood as already stated, 
there have been relatively few discussions about how to characterize SPT phases in generic mixed states. 
Some attempts in this direction include 
characterizing the topological nature by the indices based on the Uhlmann phase \cite{Uhlmann-86,Viyuela-R-M-14,Huang-A-14,Budich-D-15}, 
generalizing the MPS-based markers (e.g., projective representations, entanglement spectrum, etc.) 
to density operators via the doubled-Hilbert-space formalism \cite{Nieuwenburg-H-14,Verissimo-L-O-23,Ma-T-24}, and 
order-parameter-based approaches using the strange correlators \cite{Lee-Y-X-strange-cor-22}, the string order parameters \cite{deGroot-T-S-22}, etc.

The main challenge here consists in how to capture the topological nature in the density operators, since most of the existing approaches 
(the entanglement spectrum, the projective representations, etc.) which have been successfully applied to the pure cases 
are based on the ground-state wave functions and are not directly generalized to mixed states.   
For instance, in Ref.~\cite{deGroot-T-S-22}, de Groot et al. discuss the robustness of SPT phases of one-dimensional spin systems 
by focusing on the non-local string order parameters \cite{PerezGarcia-W-S-V-C-08,Pollmann-T-12}.  
Specifically, they showed that if the quantum channel $\mathcal{E}$ describing the system-environment coupling 
and the protecting symmetry $G$ satisfy a certain condition (the strong-symmetry condition defined below), 
the $G$-protected SPT phases in closed systems are robust against $\mathcal{E}$ in the sense that the corresponding string order parameter survives. 

On the other hand, since the computational ability of the MBQC is an inherent property of a given SPT phase as has been described above, 
it should be possible to capture its topological nature also by the performance of MBQC on it.   
In this paper, we propose to use the computational power of a given SPT phase to characterize the topological property behind. 
Specifically, we use the gate fidelity [see Eq.~\eqref{def_fidelity}] that measures the performance of MBQC to judge whether 
the phase still retains its topological properties even under decoherence.  The merit of this approach is that it is applicable 
regardless of whether the state is pure or mixed.   It also uncovers a remarkable connection between the feasibility of MBQC 
and the non-local string order parameters which have been used in assessing the robustness of SPT phases under decoherence \cite{deGroot-T-S-22} 
thereby giving a physical interpretation to the recent observation made in Ref.~\onlinecite{Raussendorf-Y-A-23} concerning 
the relation between the computational power (of a {\em pure} SPT state) and the string order parameter.

The rest of the paper is organized as follows. In Sec.~\ref{sec:model_and_method}, we quickly review the background of our research, which includes: (i) the MPS representation of the pure AKLT state 
and the MBQC on it, (ii) the symmetry conditions for the quantum channel that play vital roles, 
and (iii) the gate fidelity which quantities how accurately the gate is realized on the state in question (decohered resource states in general).  
In Sec.~\ref{sec:results}, we explicitly calculate the gate fidelity for the AKLT state and show that it consists of non-local string operators 
reminiscent of the string order parameter  
that has been introduced some time ago to capture the hidden antiferromagnetic order in the Haldane state.  
Remarkably, for the identity gate, the expression reduces precisely to the usual string order parameter, which leads us to characterize 
the open-system SPT phases using the performance of the identity gate.
We also give the explicit analytical expression of the gate fidelity for the simplest case, namely for the AKLT ground state.  

To see how our criterion works in specific cases, we numerically investigate the behavior of the gate fidelity 
in the presence of uncorrelated external noises in Sec.~\ref{sec:effect-noise}.  
To be specific, we consider the effects of several types of noises satisfying different symmetry conditions to demonstrate that 
the strong symmetry condition that is known to preserve the string order parameter \cite{deGroot-T-S-22} 
also protects the MBQC computational power for the identity gate.    
We also find more stringent symmetry conditions necessary to guarantee the high-fidelity realization of other unitary gates 
and suggest a possible richer structure in the decohered SPT phases.  
The main results are summarized in Sec.~\ref{sec:conclusion}. 
\section{AKLT state as a resource for MBQC}
\label{sec:model_and_method}

\subsection{MPS of AKLT model and its symmetry}
\label{subsec:AKLT_model}
The Hamiltonian of the Affleck-Kennedy-Lieb-Tasaki (AKLT) model coupled to a spin-$1/2$ on each end of the chain 
is given by \cite{AKLT,Affleck-K-L-T-88}:
\begin{equation}
\begin{split}
H_{\mathrm{AKLT}} = &  J^{\prime}  \bol{s}_{\text{in}} \cdot \bol{S}_1 
+  J \sum_{i=1}^{N-1}  \left[ \bol{S}_{i} \cdot \bol{S}_{i+1} +  \frac{1}{3} \left(\bol{S}_{i} \cdot \bol{S}_{i+1} \right)^2 \right]    \\
& +J^{\prime} \bol{S}_N \cdot \bol{s}_{\mathrm{out}} \; ,
\end{split}
\label{eqn:Ham-AKLT}
\end{equation} 
where $\{ \bol{S}_{i} \}$ are the spin-1s constituting the bulk and $\bol{s}_{\text{in}}$ and $\bol{s}_{\text{out}}$ are the two spin-1/2s attached 
to the ends of the chain (see Fig.~\ref{AKLT_state}). 
The reasons for adding the two terms that couple the edge spin-1/2s to the bulk spin-1s are two-fold.  
From the quantum informational viewpoint, we can initialize (in) and read-out (out) quantum information with them when implementing the MBQC.  
The condensed-matter reason will be stated later. 

This solvable model is known to share its ground-state properties with the ordinary spin-1 Heisenberg model. The unique ground state of this Hamiltonian, 
which is called the AKLT state, is conveniently given by the following (normalized and canonical) matrix product state (MPS) \cite{Perez-Garcia-V-W-C-07,SCHOLLWOCK}:
\begin{equation}
\begin{split}
&\ket{\psi_{\text{AKLT}}}   \\
&=\sum_{\sigma_{\mathrm{in/out}} } \sum_{ \{m_i \} } \left( A^{[\sigma_{\mathrm{in}}]} P^{[m_1]} A \cdots P^{[m_N]} A^{T[\sigma_{\mathrm{out}}]} \right)  \\
& ~~~~~~~~~~~~~~~~~~~~~~~~~~~~~~~~~~~~~~~~~ \ket{\sigma_{\mathrm{in}}}  \ket{m_1\cdots m_N}  \ket{\sigma_{\mathrm{out}} } ,  
\label{MPS_AKLT_state_eq}
\end{split}
\end{equation}
where $\ket{m_1\cdots m_N} := \otimes_{i} \ket{m_{i}}$ denote the basis states of the spin-1 part with $m_{i}\,(=0, \pm 1)$ being the eigenvalues of 
$S_{i}^{z}$ and $\ket{\sigma_{\text{in/out}}}$ ($\sigma_{\text{in/out}} = \pm 1/2$) are the states of the spin-1/2s attached at the ends 
(see Fig.~\ref{AKLT_state_MPS}).  

The matrix 
\eq{
A = 
\begin{pmatrix}
0 & -1/\sqrt{2}  \\
1/\sqrt{2} & 0 
\end{pmatrix} 
= - \frac{i}{\sqrt{2}} Y 
}
creates the singlet bond in Fig.~\ref{AKLT_state}, and 
\eq{
P^{[1]} = 
\begin{pmatrix}
\frac{2}{\sqrt{3}} & 0  \\
0 & 0 
\end{pmatrix}
, ~
P^{[0]} = 
\begin{pmatrix}
0 & \sqrt{\frac{2}{3}}  \\
\sqrt{\frac{2}{3}} & 0
 \end{pmatrix}
, ~
P^{[-1]} = 
\begin{pmatrix}
0 & 0  \\
0 & \frac{2}{\sqrt{3}}
\end{pmatrix}
}
projects the products of two spin-$1/2$ states on each site onto those of a spin-1.  
The exceptional matrix $A^{[\sigma]}$ ($A^{T[\sigma]}$) on the edge is the row (column) vector sliced from the matrix $A$, i.e. $A^{[\sigma]}$ is a row vector with components $A^{[\sigma]}_{\sigma'} = A_{\sigma \sigma'}$ and $A^{T[\sigma]}$ is its transpose. 
It is easy to show that the state \eqref{MPS_AKLT_state_eq} is the exact ground state of the Hamiltonian \eqref{eqn:Ham-AKLT} 
as far as $J,J^{\prime} >0$.  

The MPS form of the AKLT state \eqref{MPS_AKLT_state_eq} is diagrammatically represented as Fig.~\ref{AKLT_state_MPS}.  
It is well-known that the AKLT state in a finite open chain has four-fold degeneracy associated with the emergent free spin-1/2 degree of freedoms 
on both edges \cite{Hagiwara-K-A-H-R-90,Kennedy-90}. 
The two spin-1/2s at the edges have been introduced to suppress the four-fold degeneracy and obtain the unique singlet ground state 
without changing the physical property of the bulk \cite{White-H-93}. 

\begin{figure}[h]
\begin{center}
\includegraphics[width=\columnwidth,clip]{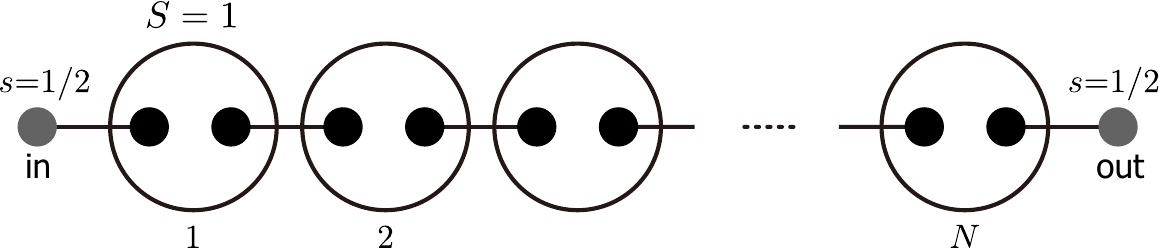}
\caption{A schematic diagram of the AKLT state \eqref{MPS_AKLT_state_eq} 
coupled with two extra one-half spins. Dots and line segments connecting them denote 
spin-1/2 degrees of freedom and singlet bonds, respectively.  Circles encircling two adjacent dots mean the projection onto the spin-1 degree of freedom. 
\label{AKLT_state}}
\end{center}
\end{figure}

\begin{figure}[h]
\begin{center}
\includegraphics[width=\columnwidth,clip]{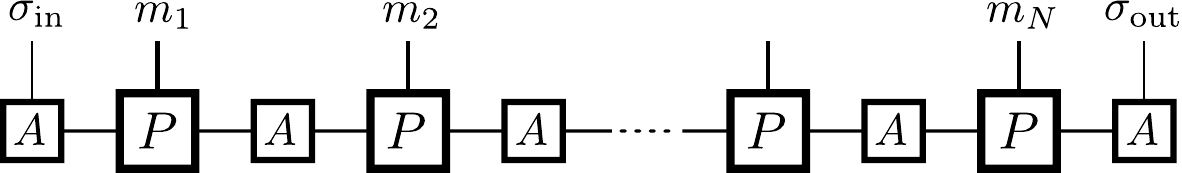}
\caption{Diagrammatic representation of MPS \eqref{MPS_AKLT_state_eq} for the AKLT state.  
Large (small) squares respectively denote the MPS tensors $P$ and $A$ in Eq.\;\eqref{MPS_AKLT_state_eq}.  
Horizontal lines connecting the adjacent squares denote entanglement bonds, 
which are introduced every time when the matrices are multiplied in Eq.\;\eqref{MPS_AKLT_state_eq}. 
The vertical lines denote the physical degrees of freedom (spin-$1/2$ and $1$ here) on each site, which are labeled by 
$\sigma_{\mathrm{in,out}}$ (for edge spin-$1/2$) and $m_i$ (for the bulk spin-$1$).  
Here, the exceptional spin-1/2 degrees of freedom on edges are described as short vertical lines.
\label{AKLT_state_MPS}}
\end{center}
\end{figure}

Generally, if $\ket{\psi}$ is symmetric under an element $g$ of a certain group $G$, then $U_g \ket{\psi} = \ket{\psi}$ up to a phase. 
In particular, when the $U_g$ is given by a tensor product of the linear representation $u_{g}(i)$ of an onsite symmetry, 
i.e., $U_g = \otimes_i u_g(i)$, then we can show that the onsite symmetry action on the physical degrees of freedom (the open circle) 
{\em fractionalizes} into those on the virtual indices (the filled squares) \cite{PerezGarcia-W-S-V-C-08}:
\begin{equation}
\begin{split}
& \sum_{n} [ u_{g}]_{mn} \Gamma^{[n]} = \be^{i \varphi_{g}} V_{g}^{\dagger} \Gamma^{[m]} V_{g}  \\
& \raisebox{-3ex}{\includegraphics[scale=0.35]{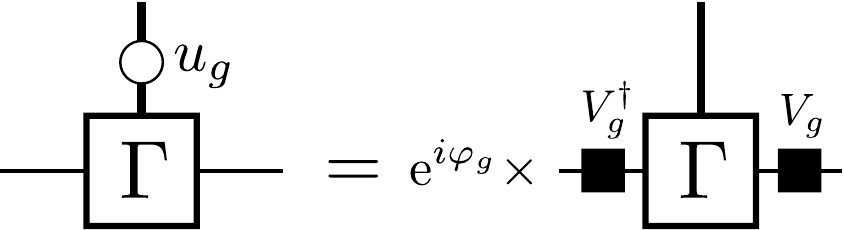}}  \;\;  .
\end{split}
\label{PG_sym} 
\end{equation}
Here, $\Gamma$ is an MPS tensor of the symmetric state $\ket{\psi}$, which satisfies 
the following canonical conditions \cite{PerezGarcia-W-S-V-C-08,Orus-V-08}: 
\begin{subequations}
\begin{align}
& \raisebox{-5ex}{\includegraphics[scale=0.3]{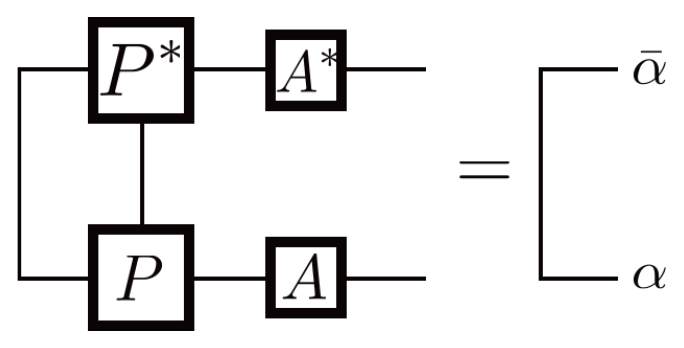}} \label{canonical_cond_left}  \quad  (= \delta_{\bar{\alpha}\alpha} ) \\
&\raisebox{-5ex}{\includegraphics[scale=0.3]{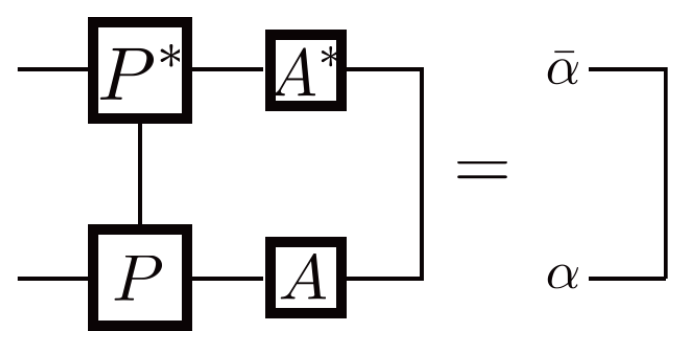}} \label{canonical_cond_right}  \quad  
 (= \delta_{\bar{\alpha}\alpha} ) \; .
\end{align}
\end{subequations}
If we apply Eq.~\eqref{PG_sym} to $u_{gh} = u_{g}u_{h}$, we immediately see that the unitary representation $V_{g}$ satisfies 
$V_{gh} = \be^{i \omega(gh)} V_{g} V_{h}$, 
i.e., $V_{g}$ is a projective representation of $G$.     
 
The MPS of the AKLT state \eqref{MPS_AKLT_state_eq} indeed satisfies the conditions \eqref{canonical_cond_left} and 
\eqref{canonical_cond_right}.  
If we are given a canonical MPS tensor and a linear representation of $G$, there is a convenient way to obtain $V_g$ 
for every $g \in G$ \cite{Pollmann-T-12}.  
Interestingly, the one-dimensional SPT phases protected by a group $G$ are classified 
by enumerating non-trivial projective representations of $G$ \cite{Chen-G-W-11,Schuch-PG-C-11,Chen-G-L-W-13}.  
The Haldane phase, to which the AKLT state belongs, is one of the simplest and most prominent nontrivial SPT phases, 
and is known to be protected by one of the following symmetries \cite{Pollmann-T-B-O-10,Pollmann-B-T-O-12}: 
the $\mathbb{Z}_2 \times \mathbb{Z}_2$ symmetry, the spatial inversion symmetry, 
and the time-reversal symmetry.  
Non-zero values of the string order parameters (SOP) \cite{string_order}
\begin{equation}
\mathcal{O}_{\mathrm{str}}^{\alpha} 
= \left\langle  S_{i}^{\alpha} \be^{ i \pi \sum_{k=i}^{j-1} S_{k}^{\alpha} } S_{j}^{\alpha} \right\rangle  \quad (\alpha=x,z) 
\label{eqn:Ostr}
\end{equation}
of the (spin-1) AKLT state imply that it is characterized by a non-trivial projective representation $V_{g}$ 
of onsite $\mathbb{Z}_2 \times \mathbb{Z}_2$ and is in the Haldane phase \cite{Pollmann-T-12,Hasebe-T-13}.   
In this sense, the SOP can be used to detect the SPT phase protected by the $\mathbb{Z}_2 \times \mathbb{Z}_2$ symmetry.

\subsection{Symmetry condition for quantum channels}
Now let us consider systems coupled to environment.   In those systems, states are generally represented by a density operator 
$\rho$ and its time evolution is described by a completely positive trace-preserving (CPTP) map or a quantum channel $\mathcal{E}$ 
(see, e.g., Refs.~\cite{Breuer-P-book-02,Nielsen-Chuang-11} for reviews):
\begin{equation}
\rho \to \rho^{\prime} = \mathcal{E}[\rho] \; . 
\end{equation}
In this subsection, we briefly review the symmetry condition satisfied by quantum channels $\mathcal{E}$ 
that preserve SPT phases \cite{deGroot-T-S-22}.
Given a symmetry group $G$ whose element $g\,(\in G)$ acts on a state $\rho$ as 
\begin{equation}
\begin{split}
& \mathcal{U}_{g}(\rho) := U_{g} \rho U_{g}^{\dagger}  \\
& \left[ \mathcal{U}^{\dagger}_{g}(\rho) := U^{\dagger}_{g} \rho U_{g} \; , \;\;  
\mathcal{U}^{\dagger}_{g} \circ \mathcal{U}_{g} = \mathcal{U}_{g} \circ \mathcal{U}^{\dagger}_{g} = 1  \right]  
\end{split}
\end{equation}
(with the symbol $\circ$ denoting the multiplication of maps),  
a quantum channel $\mathcal{E}$ is said to possess {\em weak symmetry} if it commutes with the unitary operation $\mathcal{U}_{g}$ \cite{Buca-P-12}:
\begin{subequations}
\begin{equation} 
\mathcal{U}_{g} \circ \mathcal{E} = \mathcal{E} \circ \mathcal{U}_{g} \; ,
\end{equation}
or more explicitly,
\begin{equation} 
U_{g} \mathcal{E}[\rho]  U_{g}^{\dagger} 
= \mathcal{E}\left[ U_{g} \rho U_{g}^{\dagger} \right]  \; .
\label{weak_sym_cond_original}
\end{equation}
\end{subequations}
If we use the Kraus operator-sum representation 
\begin{equation}
\mathcal{E}(\rho) = \sum_{\alpha} K_{\alpha} \rho K_{\alpha}^{\dagger}  \quad \left( \sum_{\alpha} K_{\alpha}^{\dagger}K_{\alpha} = 1  \right) 
\label{eqn:OSR}
\end{equation}
of the channel $\mathcal{E}$, the above condition \eqref{weak_sym_cond_original} can be translated into that for the Kraus operators 
$\{ K_{\alpha} \}$:
\eq{
\sum_\alpha \left( U_g K_\alpha U^\dagger_g \right) \rho \left( U_g K_\alpha U^\dagger_g \right)^\dagger = \sum_\alpha  K_\alpha \rho  K_\alpha^\dagger \;.
\label{weak_sym_cond_Kraus} }

Generally, any completely positive trace-preserving (CPTP) map $\mathcal{E}$ has ambiguity in its Kraus representation; 
in order for different Kraus representations $\{K_\alpha \}$ and $\{K'_\alpha \}$ to define the same quantum channel, 
they must be related to each other by a unitary transformation $V$ \cite{Nielsen-Chuang-11}: 
\begin{equation}
K^{\prime}_\alpha = U_g K_\alpha U^\dagger_g = \sum_\beta K_{\beta} V_{\beta \alpha}  \; .
\label{eqn:unitary-ambiguity}
\end{equation}  
Since unitary matrices are diagonalizable, Eq.~\eqref{eqn:unitary-ambiguity} immediately implies that for any $g \in G$, 
there exists a representative set $\{K^g_\alpha \}_\alpha$, among all the physically equivalent sets of the Kraus operators $\{ K_{\alpha} \}$, 
such that $v$ is diagonal:
\begin{subequations}
\begin{equation}
 U_g K^g_\alpha U^\dagger_g = \be^{i\phi_\alpha (g)} K^g_{\alpha} ~~~~\text{for } ^\forall \alpha 
 \label{weak_sym_cond_Kraus-2}
\end{equation}
or graphically [$U_{g}$ is assumed to be of the form $U_{g} = u_{g}(1)u_{g}(2) \cdots$] 
\begin{equation}
 \raisebox{-9ex}{\includegraphics[scale=0.55]{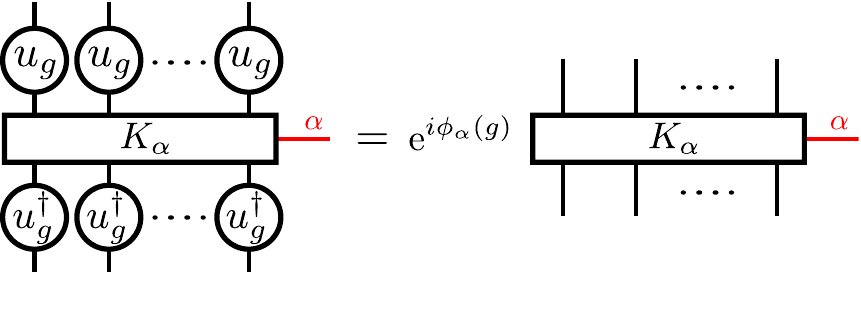}} 
 \label{weak_sym_cond_Kraus_diagram} . 
\end{equation} 
\end{subequations}
 Note that the phase factors $e^{i \phi_\alpha (g)}$ can be different for the Kraus operators $\{ K_{\alpha} \}$ and that 
the condition \eqref{weak_sym_cond_Kraus} or \eqref{weak_sym_cond_Kraus_diagram} holds for a particular $g$-dependent set 
of the Kraus operators (this is why we have put the superscript ``$g$'' to $K_{\alpha}$).  

The strong symmetry condition for the channel $\mathcal{E}$ is obtained when we further require that the phase factors $\be^{i \phi_\alpha (g)}$ 
in Eqs.\;\eqref{weak_sym_cond_Kraus} or \eqref{weak_sym_cond_Kraus_diagram} are common to all the Kraus operatrors $\{ K_{\alpha} \}$, i.e., 
when the unitary $v$ in Eq.\;\eqref{eqn:unitary-ambiguity} is a scalar matrix $\be^{i \phi (g)} \mathbf{1}$ \cite{deGroot-T-S-22}:
\begin{subequations}
\begin{equation}
\mathcal{U}_{g} (K_\alpha) 
= U_g K_\alpha U^\dagger_g = \be^{i \phi (g)} K_{\alpha} ~~~~\text{for } ^\forall \alpha 
\label{strong_sym_cond_Kraus}  \; , 
\end{equation}
which is represented diagrammatically as:
\begin{equation}
\raisebox{-9ex}{\includegraphics[scale=0.55]{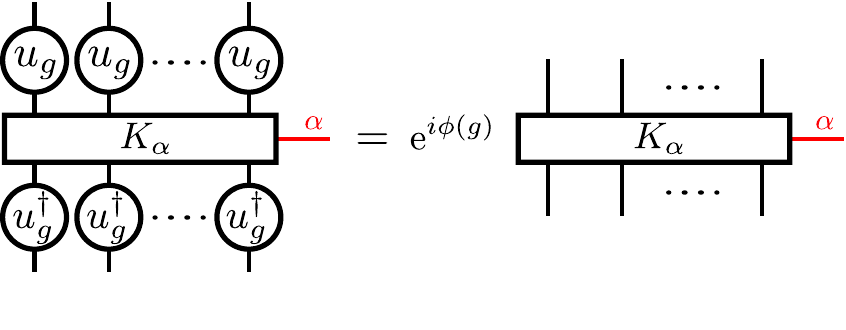}} \; .
\label{strong_sym_cond_Kraus_diagram}
\end{equation} 
\end{subequations}
Note that if a certain set of the Kraus operators satisfies \eqref{strong_sym_cond_Kraus}, so do {\em all} the others that are related by unitaries. 
Therefore, the strong-symmetry condition \eqref{strong_sym_cond_Kraus} does not depend on the particular choice of the Kraus operators. 
With the help of the completeness condition, this strong symmetry condition \eqref{strong_sym_cond_Kraus_diagram} can be 
rephrased as \cite{deGroot-T-S-22}: 
\begin{equation}
\begin{split}
\mathcal{E}^\dagger (U_g)  
& := \sum_{\alpha} K^{\dagger}_{\alpha} U_{g} K_{\alpha} = e^{i \phi (g)} U_{g}   \; , 
\intertext{i.e.,}
& \raisebox{-9.5ex}{\includegraphics[scale=0.55]{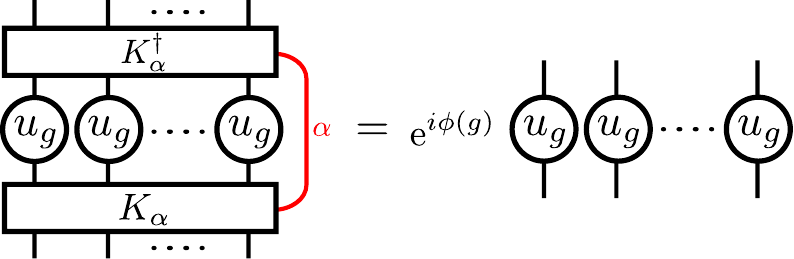}} 
\end{split}
\label{strong_sym_cond_Kraus_diagram_2}
\end{equation}
which means, in the Heisenberg picture, that after the time-evolution under $\mathcal{E}$, $U_g$ is kept invariant up to a ($g$-dependent) phase.  

 Here, an important remark is in order about the relation between the symmetry conditions 
 [Eqs.~\eqref{weak_sym_cond_Kraus} and \eqref{strong_sym_cond_Kraus}] and the choice of a (linear) representation of the on-site symmetry group $G$.  
The symmetry conditions have the definite meanings only after a particular representation of $G$ is given; 
 even if a Kraus representation satisfies the strong symmetry condition for a certain representation of $G$, it may not for other representations.  
 In Ref.~\onlinecite{deGroot-T-S-22}, it is shown that uncorrelated noises $\mathcal{E} = \circ_i \mathcal{E}_i$ being strongly symmetric with respect to $G$ 
 is necessary and sufficient for $G$-protected SPT phases to be preserved.  
 In the following section, we will translate this into the condition for the feasibility of the MBQC.


\subsection{MBQC on AKLT state and gate fidelity}
\label{subsec:fidelity}
In this subsection, we first introduce the MBQC on the AKLT state and then define the gate fidelity that quantifies 
the computational power under decoherence.  
Throughout this paper, we only consider the simplest MBQC protocol on the AKLT state introduced in Ref.~\onlinecite{Brennen-M-08},  
in which we measure the resource state in the following basis: 
\begin{subequations}
\begin{align}
& \ket{\alpha_\theta}_i := \frac{1}{\sqrt{2}} \left\{ - \be^{-i\frac{\theta}{2}} \ket{1}_i + \be^{ i\frac{\theta}{2}} \ket{-1}_i \right\} \;,  
\label{eqn:MBQC-basis-1} \\
& \ket{\beta_\theta}_i := \frac{1}{\sqrt{2}} \left\{ \be^{-i\frac{\theta}{2}} \ket{1}_i + \be^{i\frac{\theta}{2}} \ket{-1}_i \right\}   \; ,  
\label{eqn:MBQC-basis-2} \\
&\ket{\gamma}_i := \ket{0}_{i}   \; . \label{eqn:MBQC-basis-3} 
\end{align}
\end{subequations} 
The states $\ket{\alpha_\theta}_i$, $\ket{\beta_\theta}_i$, and $\ket{\gamma}_i$ 
are the zero-eigenvectors of $\tilde{S}_i^x(\theta) :=S_i^x \text{cos }(\theta /2) + S_i^y \text{sin }(\theta /2)$, 
$\tilde{S}_i^y(\theta)  :=  - S_i^x \text{sin}(\theta /2) + S_i^y \text{cos}(\theta /2)$, and $\tilde{S}_i^z(\theta)  := S_i^z$, respectively. 
For later purposes, it is convenient to introduce the following projection operators:
\begin{equation}
\mathcal{P}(m) := 
\begin{cases}
\ket{\alpha_\theta}\! \bra{\alpha_\theta}  & \text{when $m=\alpha_{\theta}$} \\
\ket{\beta_\theta}\! \bra{\beta_\theta}  & \text{when $m=\beta_{\theta}$} \\
\ket{\gamma}\! \bra{\gamma}  & \text{when $m=\gamma$} 
\end{cases} 
\label{projective_measuremnt_theta}
\end{equation}
and $\mathcal{P}^{\prime}(m)$ which is obtained from the above by setting $\theta=0$ on the right-hand side.  

 It is known \cite{Brennen-M-08} that with a single measurement in this basis, the $Z$-rotation gate $U_Z(\theta)= \be^{-i\theta Z/2}$ 
 or the identity gate is realized probabilistically unless $\theta = 0,~ \pi$.   Let us quickly review this.  
 A single projective measurement on site-$i$ results in one of the three outcomes $\ket{\alpha_\theta}_i$, $\ket{\beta_\theta}_i$, and $\ket{\gamma}_i$ 
with equal probability $1/3$.  
 If we get the result, e.g., $\ket{\alpha_\theta}_i$ for the AKLT state \eqref{MPS_AKLT_state_eq}, 
 the state after the measurement reads \cite{Gross-E-07,Brennen-M-08}:
\begin{equation}
\begin{split}
&\ketbra{\alpha_\theta}{\alpha_\theta}_i |\psi_{\text{AKLT}} \rangle    \\
&=\sum_{\sigma_{\text{in/out}} } \sum_{ \{m_k |k \neq i\} } A^{[\sigma_{\text{in}}]} P^{[m_1]} \cdots A \left(\sum_{m_i} \braket{\alpha_\theta}{m_i } P^{[m_i]} A \right)   \\
& \cdots  P^{[m_N]} A^{T[\sigma_{\text{out}}]} \ket{\sigma_{\text{in}}}  \ket{m_1\cdots m_{i-1} \alpha_\theta m_{i+1} ~ \cdots m_N}  \ket{\sigma_{\text{out}} } 
\end{split}
\label{MPS_AKLT_state_single_measurement}
\end{equation}
with the quantity inside the bracket given by:
\begin{equation}
\raisebox{-1.6ex}{\includegraphics[scale=0.3]{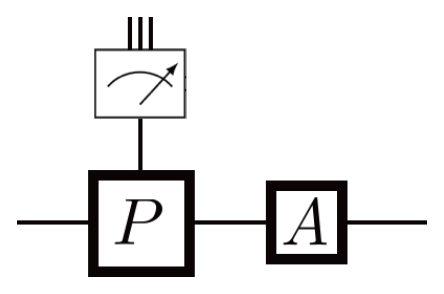}}  = \left \{
\begin{array}{l}
\sum_{m_i} \braket{\alpha_\theta}{m_i } P^{[m_i]}A~ \propto Xe^{-i\theta Z /2} \\
\\
\sum_{m_i} \braket{\beta_\theta}{m_i } P^{[m_i]}A~ \propto XZe^{-i\theta Z /2} \\
\\
\sum_{m_i} \braket{\gamma}{m_i } P^{[m_i]}A~ \propto Z  \; .
\end{array}
\right. 
\label{single_measurement_result}
\end{equation}
Therefore, for a given AKLT state on a length-$N$ spin-1 chain, the implementation of the $Z$-rotation gate proceeds as follows:  
First, we perform measurements in the basis \eqref{eqn:MBQC-basis-1}-\eqref{eqn:MBQC-basis-3} with the measurement angle $\theta$ 
until either the result $\ket{\alpha_\theta}$ or $\ket{\beta_\theta}$ (the ``desired'' results) is obtained.  When we obtain either of these two at a certain site, 
we reset $\theta$ to 0, and measure all the remaining sites with $\theta=0$.  
This ensures that the probability of successfully implementing the $Z$-rotation gate is generally given by $1-(1/3)^N$.   
Moreover, the obtained $Z$-rotation gate $\be^{-i\theta Z /2}$ includes (i) $I~(\theta =0)$, (ii) $Z~(\theta =\pi)$, (iii) $S~(\theta =\pi/2)$, 
(iv) $T~(\theta =\pi /4)$ gates, and as already noted, if $\theta = 0$, the MBQC on the pure AKLT state is deterministic.  
Now it is clear that if we use, instead of \eqref{eqn:MBQC-basis-1}-\eqref{eqn:MBQC-basis-3}, the basis rotated along the $X$-axis, 
we can implement the $X$-rotation gate \cite{Brennen-M-08}.   

Next, we introduce the gate fidelity, which quantifies how accurately the MBQC is done. It is defined as the fidelity between 
 two resource states $\rho_U$ and $\ketbra{\psi_U}{\psi_U}_{\text{res}}$ of a two-qubit system made of the input (``in'') and output (``out'') 
 (see Fig.~\ref{AKLT_state}): 
\eq{F_U =\text{Tr}_{\text{in}, \text{out}}\left( \rho_U \ketbra{\psi_U}{\psi_U}_{\text{res}} \right) 
\label{def_fidelity} \; }
Here, $\ket{\psi_U}_{\text{res}}$ denotes the {\em ideal} resource state (defined below) for the gate $U$ teleportation \cite{Gottesman-C-99} 
and the trace $\text{Tr}_{\text{in}, \text{out}}$ is over the two-qubit states. 

On the other hand, $\rho_U$ is the resource state which may have been 
decohered by some noise.   To be specific, it is defined as the mixed state we get after consecutive measurements are done 
{\em on the decohered state} $\tilde{\rho}$ of the entire (i.e., consisting of both spin-$1$ and $1/2$) system: 
\begin{equation}
\rho_U = \text{Tr} \left[ \sum_{\boldsymbol{m}} B_{\text{out}}^{(\boldsymbol{m} )} \boldsymbol{\mathcal{P}}(\boldsymbol{m}) \tilde{\rho}
 ~\boldsymbol{\mathcal{P}}(\boldsymbol{m}) B_{\text{out}}^{ ( \boldsymbol{m} ) \dagger} \right]   \; .
\label{definition_rho_U}
\end{equation}
Here, $\boldsymbol{m} =(m_1,\dots,m_N )$ is a set of measurement outcomes and $B_{\text{out}}^{(\boldsymbol{m} )}$ 
is the corresponding Pauli by-product operator.  In contrast to $F_{U}$, now the partial trace is taken only over the spin-1 subsystem in the bulk. 
Also, we have introduced the following short-hand notation for the consecutive projection operators: 
\begin{equation}
\begin{split}
& \underbrace{ 
\mathcal{P}'_{N}(m_N)\cdots \mathcal{P}'_{l(\boldsymbol{m})+1}(m_{l(\boldsymbol{m}) +1}) 
}_{\theta=0}
\mathcal{P}_{l(\boldsymbol{m})}(m_{l(\boldsymbol{m}) }) \cdots \mathcal{P}_{1}(m_1)  \\
& = \boldsymbol{\mathcal{P}}(\boldsymbol{m}) \; , 
\end{split}
\end{equation}
where $l(\boldsymbol{m})$ [$1\leq  l(\boldsymbol{m}) \leq N$] represents the position of the site at which one of the ``desired results'' (i.e., $\ket{\alpha_\theta}$ and $\ket{\beta_\theta}$) is obtained for the first time 
and is uniquely determined by a set $\boldsymbol{m}$ of the measurement outcomes. 
The ``desired results'' depend on the targeted gate and are, for example, $\ket{\alpha_\theta}$ and $\ket{\beta_\theta}$ 
for the $Z$-rotation gate $U = U_{Z} (\theta) \, (= \be^{-i \theta Z /2})$.  
The measurement basis are the same as those for the MBQC on the pure AKLT state.   
 
As has been discussed in Sec.~\ref{sec:intro}, we characterize the topological nature in the mixed-state SPT phases by 
the MBQC computational power.  
To this end, we first pick up a particular gate $U$ to calculate $F_{U}$ for the SPT ground state (e.g., the pure AKLT state),  
and then compare this with the value obtained for the corresponding decohered state $\tilde{\rho}$.  
When the two coincide, we regard the decohered SPT state as still retaining 
the computational ability for the gate \footnote{%
Of course, this definition might be too strict for the practical implementation.  We have introduced 
this stringent definition solely for clearly defining the stability of SPT phases in the mixed-state settings.}.  
In particular, for the reason which will become clear below 
[see Eq.~\eqref{identity_fidelity}], we focus on the fidelity for the identity gate $U=I$ as the most general indicator. 
That is, we say a given $G$-protected SPT phase is robust against the decoherence when $F_{I}$ for the pure SPT state and 
the decohered one are equal. 

\begin{widetext}
Now we proceed to calculating the fidelity of the $Z$-rotation gate $U_{Z} (\theta) = \be^{-i \theta Z /2}$.   
The ideal resource state for $U_{Z} (\theta)$ is 
$\ket{\psi_{U_Z(\theta)}}_{\text{res}} 
=\frac{1}{\sqrt{2}} \left( \ket{0}_{\text{in}} \ket{0}_{\text{out}} +  e^{i\theta} \ket{1}_{\text{in}} \ket{1}_{\text{out}}  \right)$. Therefore, using its stabilizer form \cite{Nielsen-Chuang-11}, we can write the fidelity $F_{U}$ of the $U_Z(\theta)$ gate as:
\begin{equation}
\begin{split}
F_{U_Z(\theta)} = & \text{Tr}_{\text{in, out}}\left[\rho_{U_Z(\theta)} 
\left(\frac{I_{\text{in}}I_{\text{out}}  + Z_{\text{in}}Z_{\text{out}} }{2} \right) 
\left( \frac{I_{\text{in}}I_{\text{out}}  + X_{\text{in}}\, \be^{-i\theta Z_{\text{out}}/2} X_{\text{out}}\,  \be^{i\theta Z_{\text{out}}/2} }{2} \right) \right]  \\
= & \frac{1}{4}  \text{Tr}_{\text{in, out}}\left[\rho_{U_Z(\theta)} \right] 
+  \frac{1}{4} \text{Tr}_{\text{in, out}}\left[\rho_{U_Z(\theta)}   Z_{\text{in}}Z_{\text{out}}  \right]
+ \frac{1}{4}  \text{Tr}_{\text{in, out}}\left[\rho_{U_Z(\theta)}   X_{\text{in}}\, \be^{-i\theta Z_{\text{out}}/2} X_{\text{out}}\, \be^{i\theta Z_{\text{out}}/2} \right]  \\
& + \frac{1}{4}  \text{Tr}_{\text{in, out}} \left[\rho_{U_Z(\theta)} 
X_{\text{in}}Z_{\text{in}}\,  \be^{-i\theta Z_{\text{out}}/2} X_{\text{out}} Z_{\text{out}} \,  \be^{i\theta Z_{\text{out}}/2} 
  \right]   \; . 
\end{split}
\label{gate_fidelity_z_rot}  
\end{equation}
\end{widetext}
For the case where $\tilde{\rho}$ in Eq.~\eqref{definition_rho_U} is maximally mixed state corresponding to the infinitely high temperature, 
$F_{U_Z(\theta)}=1/4$ for any $\theta$, so $F_U = 1/4$ means that we fail to implement the unitary gate $U$ at all.

\section{Gate fidelity of AKLT state}
\label{sec:results}

In this section, we explicitly calculate the gate fidelity \eqref{gate_fidelity_z_rot} to show that 
this quantum-information-originated quantity is closely related to the physical string-order parameters \eqref{eqn:Ostr}.  
Specifically, depending on the unitary gate $U$ under consideration, different string order parameters contribute to the gate fidelity $F_U$. 


\subsection{Calculation of the gate fidelity of $Z$-rotation}
\label{calculation_of_the_gate_fidelity}
The gate fidelity \eqref{gate_fidelity_z_rot} for the $Z$-rotation gate $U_{Z} (\theta) = \be^{-i \theta Z /2}$ 
consists of four terms, each of which can be calculated exactly.  
The first of them which does not contain the Pauli operators just yields a trivial constant term 
$\frac{1}{4}\text{Tr}_{\text{in, out}}\left[\rho_{U_Z(\theta)} \right] =1/4$ in the gate fidelity.   
Therefore, all the non-trivial contributions come from the other three terms which we focus on in the following. 

First, we show that the second term $\text{Tr}_{\text{in, out}}\left[\rho_{U_Z(\theta)} Z_{\text{in}} Z_{\text{out}} \right]$ is explicitly calculated as:
\begin{equation}
\text{Tr}_{\text{in, out}}\left[\rho_{U_Z(\theta)} Z_{\text{in}} Z_{\text{out}} \right] 
= - \widetilde{\text{Tr} }  \left[  \tilde{\rho} \, Z_{\text{in}} \left( \prod_{j=1}^{N} \be^{i\pi S_j^z} \right) Z_{\text{out}} \right]  \; ,
\label{fidelity_ZZ_term}
\end{equation}

where $\tilde{\rho}$ is the density operator of the composite system made of the spin-1 bulk and the two spin-1/2s at the edges,  
and the trace $\widetilde{\text{Tr} }:= \text{Tr}_{\text{in, out}} \text{Tr}$ is taken over this composite system.  
The right-hand side of \eqref{fidelity_ZZ_term} is closely related to the ordinary string order parameter \cite{string_order,Garcia-W-S-V-C-08}
associated with the on-site $\mathbb{Z}_2\times \mathbb{Z}_2$-symmetry. 
The calculation of the $Z_{\text{in}} Z_{\text{out}}$-term in the gate fidelity \eqref{gate_fidelity_z_rot} goes 
as in the case of the MBQC with the cluster state \cite{Fujii-N-O-M-13}.  
First, by the definition \eqref{definition_rho_U} of $\rho_{U}$ and the cyclic property of the trace, the original expression can be recast as:
\begin{equation}
\begin{split}
&\displaystyle \text{Tr}_{\text{in, out}}\left[ \rho_{U_Z(\theta)} Z_{\text{in}} Z_{\text{out}} \right]   \\ 
&= \text{Tr}_{\text{in, out}} \left\{ 
\text{Tr} \sum_{\boldsymbol{m}} \left[ B^{(\boldsymbol{m})}_{\text{out}} \boldsymbol{\mathcal{P}}(\boldsymbol{m}) \tilde{\rho} \,  
\boldsymbol{\mathcal{P}} (\boldsymbol{m}) B^{(\boldsymbol{m}) \dagger }_{\text{out}}  Z_{\text{in}} Z_{\text{out}} \right]  \right\}  \\
&=  \widetilde{\text{Tr} }\sum_{\boldsymbol{m}} \left[ 
\tilde{\rho} \,  Z_{\text{in}}  \boldsymbol{\mathcal{P}} (\boldsymbol{m}) B^{(\boldsymbol{m}) \dagger }_{\text{out}}  
Z_{\text{out}}  B^{(\boldsymbol{m})}_{\text{out}} \boldsymbol{\mathcal{P}} (\boldsymbol{m})  \right]  \; .
\end{split}
\label{calculation_of_ZZ_term_1}
\end{equation}
As is shown in Appendix~\ref{sec:deriv-FU}, we can eliminate the by-product operators just leaving a string of the $\mathbb{Z}_{2}$ generators:
\begin{equation}
B^{(\boldsymbol{m}) \dagger }_{\text{out}}  Z_{\text{out}}  B^{(\boldsymbol{m})}_{\text{out}} \boldsymbol{\mathcal{P}}(\boldsymbol{m}) 
= - \boldsymbol{\mathcal{P}}(\boldsymbol{m}) \left( \prod_{j=1}^{N} \be^{i\pi S_j^z}  \right) Z_{\text{out}} \; .
\label{eqn:BZB-by-string}
\end{equation}
Plugging this into Eq.~\eqref{calculation_of_ZZ_term_1} and 
carrying out the summation over all the possible measurement outcomes $\{ \boldsymbol{m} \}$, we obtain the desired result \eqref{fidelity_ZZ_term} 
(see Appendix~\ref{sec:deriv-FU} for more details).

\begin{widetext}
The other two terms can be calculated similarly.  
The only difference is that now we need two different string operators 
$\left(\prod_{j=1}^{l(\boldsymbol{m})} \be^{i \pi \tilde{S}_j^{x,y}(\theta)  } \right)$ ({\em twisted}) 
and $\left(\prod_{j=l(\boldsymbol{m})+1}^{N} \be^{i \pi S_j^{x,y} } \right)$ ({\em untwisted})
to write down the expressions similar to Eq.~\eqref{eqn:BZB-by-string}.   
As a result, we obtain:
\begin{equation}
\begin{split}
&\displaystyle{\text{Tr}_{\text{in, out}}}  
\left[\rho_{U_Z(\theta)} X_{\text{in}}\,  \be^{- i \theta Z_{\text{out}} /2} X_{\text{out}}\, \be^{i \theta Z_{\text{out}} /2} \right]
 +\text{Tr}_{\text{in, out}} \left[\rho_{U_Z(\theta)} 
X_{\text{in}}Z_{\text{in}} \,  \be^{-i\theta Z_{\text{out}}/2} X_{\text{out}} Z_{\text{out}}\,  \be^{i\theta Z_{\text{out}}/2}   \right]  \\
&= - \widetilde{\text{Tr}} \sum_{\boldsymbol{m}} \left[ \tilde{\rho} \Bigl(X_{\text{in}}\text{cos} \theta + Y_{\text{in}} \text{sin} \theta \Bigr) 
\boldsymbol{\mathcal{P}}(\boldsymbol{m}) 
\left(\prod_{j=1}^{l(\boldsymbol{m})} \be^{i \pi \tilde{S}_j^x(\theta)  } \right) \left(\prod_{j=l(\boldsymbol{m})+1}^{N} \be^{i \pi S_j^x } \right)  
X_{\text{out}}  \right]  \\
& \phantom{=} 
- \widetilde{\text{Tr}} \sum_{\boldsymbol{m}} \left[ \tilde{\rho} \Bigl(-X_{\text{in}}\text{sin} \theta + Y_{\text{in}} \text{cos} \theta \Bigr) 
\boldsymbol{\mathcal{P}}(\boldsymbol{m}) \left(\prod_{j=1}^{l(\boldsymbol{m})} \be^{i \pi \tilde{S}_j^y (\theta) } \right) 
\left(\prod_{j=l(\boldsymbol{m})+1}^{N} \be^{i \pi S_j^y } \right) Y_{\text{out}}  \right] 
\; .
\label{eqn:3rd-4th-tem-in-FU}
\end{split}
\end{equation}
\end{widetext}
It is important to note that, in contrast to the term $\text{Tr}_{\text{in, out}}\left[\rho_{U_Z(\theta)} Z_{\text{in}} Z_{\text{out}} \right]$, 
this is not of the form of the {\em ordinary} string order parameter in the sense that the operator strings consist of the two different elements 
of $\mathbb{Z}_{2}$ [i.e., the twisted ($\be^{i \pi \tilde{S}_j^a (\theta)  }$) and untwisted ($\be^{i \pi S_j^a }$) ones).  

When $\theta=0$ at which $\tilde{S}^x_i (0)= S^x_i$ and $\tilde{S}^y_i (0)= S^y_i$, the twist is absent and the fidelity for the identity gate 
$U_{Z}(\theta=0)= I$ reduces to the sum of the three string order parameters corresponding to the representation 
$\{  1,~\be^{i\pi S^x},~\be^{i\pi S^y},~\be^{i\pi S^z} \}$ of $\mathbb{Z}_{2} \times \mathbb{Z}_{2}$: 
\begin{subequations}
\begin{equation}
F_I = \frac{1}{4} -  \frac{1}{4} \sum_{a=1,2,3} 
\widetilde{\text{Tr}} \left[ \tilde{\rho} \, \mathcal{X}^{(a)}_{\text{in}} \left( \prod_{j=1}^{N} \be^{i\pi S_j^a} \right) \mathcal{X}^{(a)}_{\text{out}} \right]  
\label{identity_fidelity}
\end{equation}
with $\mathcal{X}^{(a)}_{\text{in/out}}$ and $S_j^a$ defined by:
\begin{equation}
(\mathcal{X}^{(a)}_{\text{in/out}} , S_j^a ) = 
\begin{cases}
( X_{\text{in/out}} , S_j^x) & \text{for $a=1$} \\
( Y_{\text{in/out}} , S_j^y) & \text{for $a=2$} \\
( Z_{\text{in/out}} , S_j^z) & \text{for $a=3$} \; .
\end{cases}
\label{mathcal_x}
\end{equation}
\end{subequations} 

Now it is clear why the identity gate $I$ takes the special position among all the gates $U$; when $\tilde{\rho}$ is the pure state, 
the $F_{I}$ reduces to the usual string order parameter whose topological meaning is well established. 

In general, the fidelity for the $Z$-rotation gate $U_{Z}(\theta)= \be^{-i \theta Z/2}$ contains the string order parameter as well as 
other non-local string-like operators.   
It is straightforward to generalize the above to the $X$-rotation gate.  
The fidelity of the $X$-rotation gate consists of similar string operators obtained by $Z\to X$, and $S^z \to S^x$.  
Moreover, this is the case for {\em any} unitary gate $U$ since an arbitrary $U$ can be described as a rotation about a certain axis; once the axis is given, we can fix a particular ``canonical'' representation of $\mathbb{Z}_2 \times \mathbb{Z}_2$ for this axis to repeat the same steps.  
The net outcome is the expression of the gate fidelity consisting of the string operators associated to this representation. 

\subsection{Ground-state gate fidelity}
In this section, we proceed with the calculation in the previous subsection when $\rho$ is pure: 
$\rho= \ket{\psi} \!\!\bra{\psi}$ with $\ket{\psi}$ being the $G$-symmetric AKLT state $| \psi_{\text{AKLT}}\rangle$ 
[Eq.~\eqref{MPS_AKLT_state_eq}] satisfying Eq.~\eqref{PG_sym}.    
First, the $\text{Tr}_{\text{in, out}}\left[\rho_{U_Z(\theta)} Z_{\text{in}} Z_{\text{out}} \right]$ term can have a finite value 
even in the limit of an infinite chain $N \to \infty$ and is calculated in the same way as the ordinary string order parameter.  
If we use Eq.~\eqref{PG_sym} and the fact that when $u_{g} = \be^{i \pi S^{z}}$, $V_g=V^\dagger_g = Z$ in the AKLT state, 
the right-hand side of Eq.~\eqref{fidelity_ZZ_term} may be calculated as:
\begin{equation}
\begin{split}
&- \raisebox{-6ex}{\includegraphics[scale=0.35]{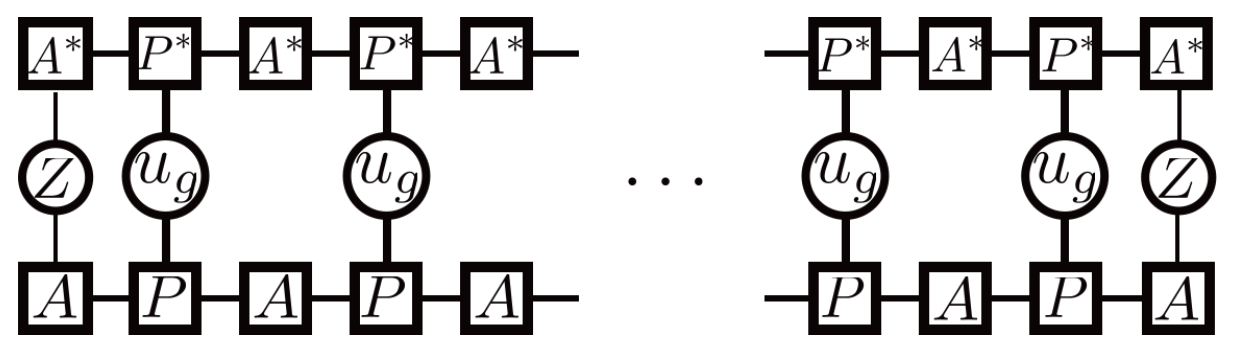}}   \\
& \Large{=} - \raisebox{-4.2ex}{\includegraphics[scale=0.35]{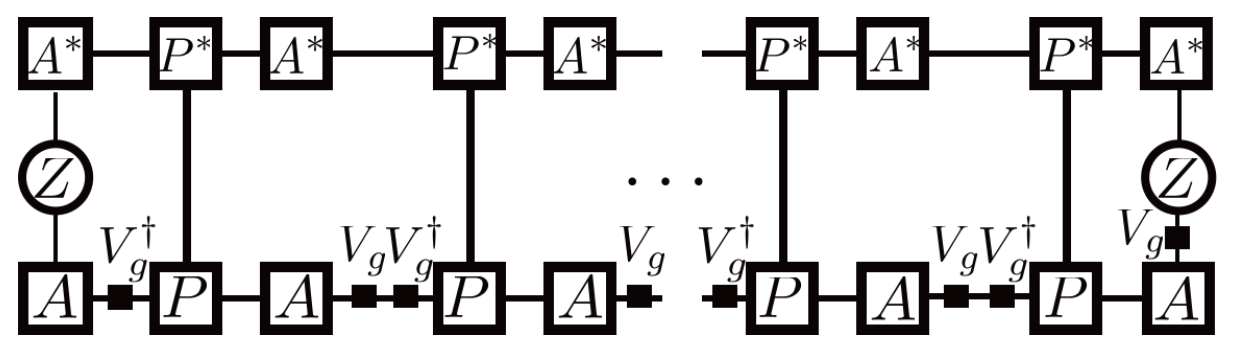}}   \\
& \Large{=} - \raisebox{-5ex}{\includegraphics[scale=0.35]{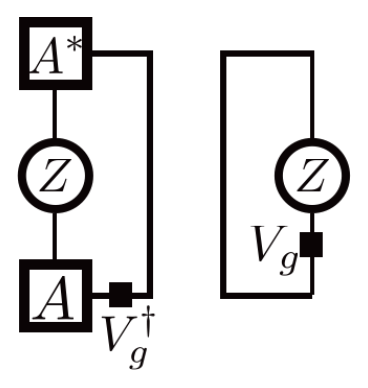}}   \\
& \Large{=} 1  \; .
\label{ZZ_term_gs}
\end{split}
\end{equation} 
In the last equality of Eq.~\eqref{ZZ_term_gs}, we have used $A \propto Y$ and $Z V_{g}=1$.  
In fact, this type of non-local correlation functions which have spin-$1/2$ operators $\{ X,Y,Z \}$ at the ends satisfy the same selection rule  
as the ordinary string order parameters \cite{Pollmann-T-12} 
obtained by replacing the end operators as $I_2\to I_3$, $X\to S^x$, $Y\to S^y$ and $Z\to S^z$ for the case of the AKLT state.

When considering the gate fidelity for general $\theta$, we must carefully deal with the remaining part \eqref{eqn:3rd-4th-tem-in-FU} [i.e., \eqref{XX_term} and \eqref{XZ_term}] for 
$\rho=\ketbra{ \psi_{\text{AKLT}} }{\psi_{\text{AKLT}}}$, 
whereas they do not reduce to the string order parameters 
since the summation over $\{ \bolm \}$ cannot be carried out explicitly and $\bolP(\bolm)$ remains in the expressions 
[see Eqs.~\eqref{XX_term} and \eqref{XZ_term}].  
Except when we obtain the measurement outcomes $\boldsymbol{m} = (\gamma,\gamma,\dots,\gamma)$, 
each term in the first line on the right-hand side of Eq.~\eqref{XX_term} is represented diagrammatically in Fig.\;\ref{typical_xx_term}.  
\begin{figure}[htb]
\begin{center}
\includegraphics[width=\columnwidth,clip]{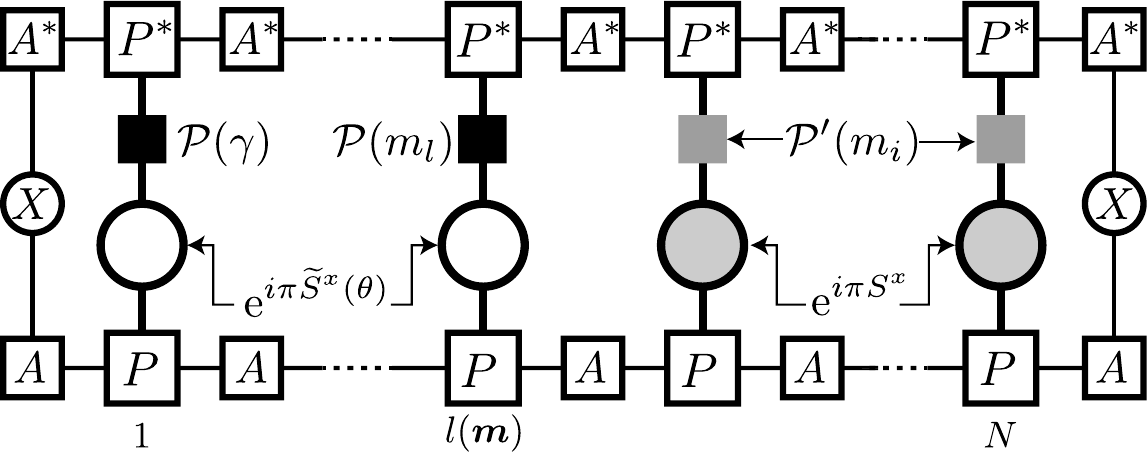}
\caption{
A typical term that appears in the first line on the right-hand side of Eq.~\eqref{XX_term}. 
The symbols $\raisebox{-1.3ex}{\includegraphics[scale=0.3,clip]{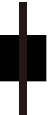}}$ 
and $\raisebox{-1.3ex}{\includegraphics[scale=0.3,clip]{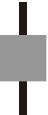}}$ denote the application of 
the projective measurement \eqref{projective_measuremnt_theta} with generic $\theta$ and $\theta=0$, respectively.
\label{typical_xx_term}}
\end{center}
\end{figure}

Despite its looking, the diagram shown in Fig.~\ref{typical_xx_term} is calculated simply as: 
\begin{equation}
 \text{Fig.\;\ref{typical_xx_term}} = \left(\frac{1}{3}\right)^{N-1} \times \;\; \raisebox{-7ex}{\includegraphics[scale=0.35,clip]{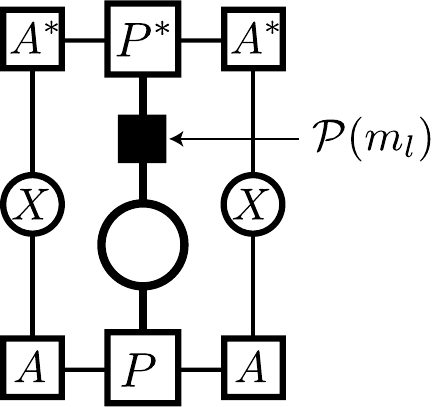}} 
= - \left(\frac{1}{3}\right)^{N} \cos \theta \;. 
\label{xx_term_example}
\end{equation} 
The first equality of Eq.\;\eqref{xx_term_example} holds since both transfer matrices
\begin{equation}
\raisebox{-7ex}{\includegraphics[scale=0.35,clip]{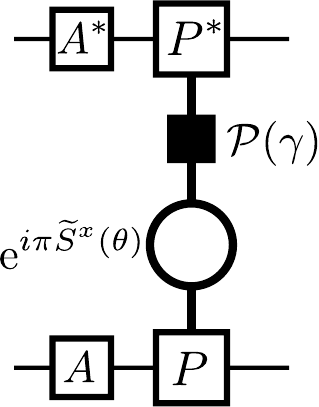}}
\quad 
\text{and} 
\quad 
\raisebox{-7ex}{\includegraphics[scale=0.35,clip]{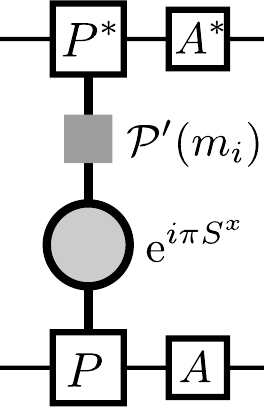}} 
\end{equation} 
respectively have the left and right eigenvectors $\raisebox{-1.3ex}{\includegraphics[scale=0.3,clip]{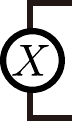}}$ 
and $\raisebox{-1.3ex}{\includegraphics[scale=0.3,clip]{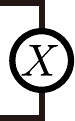}}$ with the same eigenvalue $1/3$, 
while the second equality holds since
\begin{eqnarray}
\raisebox{-7ex}{\includegraphics[scale=0.35,clip]{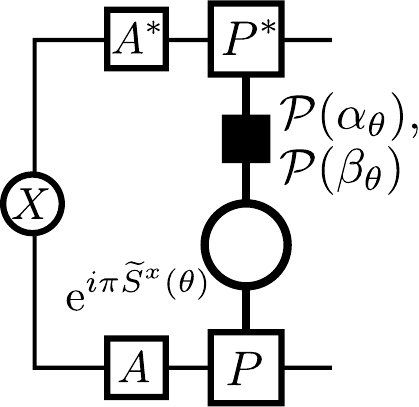}} \; 
= \frac{\text{cos}\theta}{3} \;  \raisebox{-3ex}{\includegraphics[scale=0.4,clip]{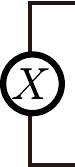}}\; 
- \frac{\text{sin}\theta}{3} \; \raisebox{-3ex}{\includegraphics[scale=0.4,clip]{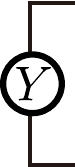}} \;.
\end{eqnarray} 
On the other hand, when we obtain the result $\boldsymbol{m} = (\gamma,\gamma,\dots,\gamma)$,
\begin{equation}
- \;   \raisebox{-8ex}{\includegraphics[scale=0.35]{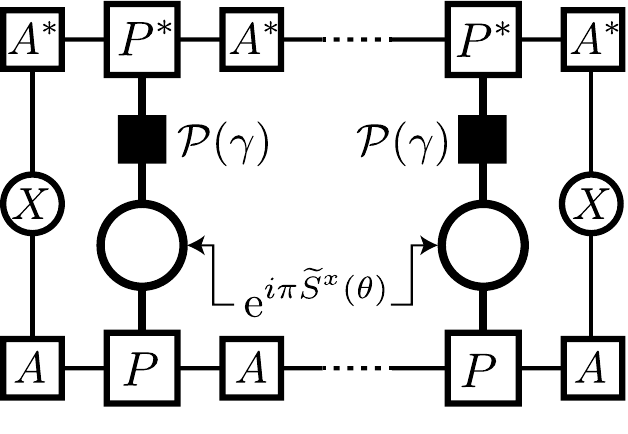}}  \; \; = \left(\frac{1}{3}\right)^N \; . 
\end{equation} 
Therefore, when the summation over $\boldsymbol{m}$ is carried out, the first line of the right-hand side of 
Eq.~\eqref{XX_term} is $\{ ( 1-\frac{1}{3^N} )\cos \theta  +\frac{1}{3^N} \} \cos \theta$.   
Then, the sum of \eqref{XX_term} and \eqref{XZ_term} is given by:
\begin{eqnarray}
2 \left(1-\frac{1}{3^N}\right)+\frac{2}{3^N} \cos \theta \;.
\end{eqnarray}
Consequently, the fidelity \eqref{gate_fidelity_z_rot} for the $Z$-rotation gate on the pure AKLT state $\rho=\ketbra{ \psi_{\text{AKLT}} }$ reads:  
\eq{ F_{U_{Z}(\theta)} = 1-\frac{1}{2\cdot 3^N} (1-\text{cos}\theta) \; .   \label{fidelity_result_gs}} 
Note that the value becomes $\theta$-independent in the limit of an infinite chain: $N \to \infty$. 

Alternatively, the same result could have been obtained in another simpler way.  We first note
\eq{
\rho_U = \left( 1- \frac{1}{3^N} \right) \ketbra{\psi_U}{\psi_U} + \frac{1}{3^N} \ketbra{\psi_I}{\psi_I} 
}
for $\rho = \ketbra{\psi_{\text{AKLT}}}{\psi_{\text{AKLT}}}$, since always the ideal resource state $\ket{\psi_U}$ is obtained 
except for the case $\boldsymbol{m} = (\gamma,\gamma,\dots,\gamma)$ which occurs with a probability $1/3^N$. 
By substituting this into the definition of the gate fidelity \eqref{def_fidelity}, we reproduce the result Eq.~\eqref{fidelity_result_gs}.  
Therefore, the fidelity of the MBQC on the pure AKLT state in general is smaller than $1$ when the system size $N$ is finite, and 
in the limit $N \to \infty$, it approaches 1 for any gate $U$ since, for any $U$, $\rho_U \to \ketbra{\psi_U}{\psi_U}$ in this limit. 
The fact that the gate fidelity of the MBQC on the pure AKLT state can generally be less than 1 for finite system sizes 
reflects that the MBQC is not deterministic but probabilistic (as is suggested by the term $1/3^{N}$) except for $\theta = 0,~\pi$. 
However, that we have obtained $F_{U}$ smaller than $1$ here must be distinguished from the reduction due to external noises,  
since we can completely trace the probabilistic failures in the MBQC on the pure AKLT state.


\section{Effects of noise on gate fidelity}
\label{sec:effect-noise}
In the previous section, we calculated the fidelity of $Z$-rotation gate with respect to a general state, and specifically, we thoroughly discussed the gate fidelity in the ground state i.e. in the case without noises.
In this section, we examine the behavior of the fidelity under various specific noises to show that our results are consistent
with the known ones \cite{deGroot-T-S-22}. 
Furthermore, the gate fidelity gives us more detailed information on the structure of decohered SPT phases, compared to the existing methods based on (non-local) order parameters. 
We focus on the site-wise uncorrelated noise $\mathcal{E} = \circ_i \mathcal{E}_i $ below. 
Our numerical calculations used the ITensor library \cite{ITensor}.

\subsection{Symmetry condition and gate fidelity}\label{sec:symm_cond_fidelity}
Before discussing the open quantum systems subject to some specific noises, we introduce a very useful proposition 
that relates the performance of MBQC to the strong symmetry condition: 
\begin{prop} \label{SS_cond_and_identity_fidelity}
The fidelity of the identity gate does not decay by an uncorrelated noise, i.e., $F_I =1$ on the noisy AKLT state 
if and only if the noise $\mathcal{E}$ satisfies the strong symmetry condition \eqref{strong_sym_cond_Kraus} 
for the canonical (linear) representation $\{ 1, \, \be^{i\pi S^x}, \, \be^{i\pi S^y}, \, \be^{i\pi S^z}  \}$ of 
the on-site $\mathbb{Z}_2 \times \mathbb{Z}_2$-symmetry.
\end{prop}
A remark is on order here about the choice of the $\mathbb{Z}_2 \times \mathbb{Z}_2$ representation.   
Since the AKLT state is isotropic, one may think that the same statement holds for any other representations that are related to the canonical one 
by rotation after due modification of the Kraus operators.  However, the measurement basis \eqref{eqn:MBQC-basis-1}-\eqref{eqn:MBQC-basis-3} 
already fix the directions and the canonical representation is crucial for the statement.   

We can easily check the above statement.  
By Eqs.~\eqref{strong_sym_cond_Kraus_diagram_2} and \eqref{identity_fidelity}, if a quantum noise satisfies the strong symmetry condition, then the gate fidelity does not decay. Conversely, if $F_I=1$ on a noisy AKLT state, each of the four terms in Eq.\;\eqref{identity_fidelity} must take its maximal value,  
which is possible only when the strong symmetry condition holds. The fact that we can perfectly perform the identity gate even on the noisy AKLT state under the strongly symmetric channel can be checked by directly probing the MPO tensor of the mixed state. (See Sec.\;\ref{summary_sec_4} and Appendix\;\ref{sec:MPO_tensor_product_decomposition})

The Prop.~\ref{SS_cond_and_identity_fidelity} suggests that the strong symmetry condition is not sufficient for the noisy AKLT state to be a one-qubit universal resource of the MBQC necessitating stronger conditions for this purpose. 
In the following section, by demonstrating the numerical results, we will see that the strong $\mathbb{Z}_2 \times \mathbb{Z}_2$ symmetry of quantum noises is not sufficient for the one-qubit universal computation.

\subsection{Discrete time-evolution of the fidelity of $Z$-rotation gate under some symmetric noises} \label{Z_rotation_fidelity_and_noise}
\label{sec:numerics}
In this section, we consider the discrete time-evolution of the SPT states under the repeated application of 
quantum noises described by \eqref{eqn:OSR}.  

Precisely, we calculate the gate fidelity $F_{U}$ at each time step 
by replacing $\tilde{\rho}$ in \eqref{definition_rho_U} with the one obtained after the step. 
First, by calculating the gate fidelities against noises shown in Table~\ref{noises_table}, we numerically show that time-reversal symmetry of the decohered state does not affect the fidelity of the $Z$-rotation gate. Then, we will see, as mentioned in the previous section, 
that the strong symmetry condition \eqref{strong_sym_cond_Kraus} or \eqref{strong_sym_cond_Kraus_diagram_2} alone 
is not sufficient for the implementation of the $Z$-rotation gate, not to mention, realization of the one-qubit universal computation. We will derive the condition to implement the $Z$-rotation gate and prove that any non-trivial quantum channel disables the one-qubit universal computation in our MBQC scheme.
\begin{center}
\begin{table}[htb]
\caption{Several quantum noises considered and the corresponding Kraus operators $\{ \mathcal{K}_{\alpha} \}$ [see Eq.~\eqref{Kraus_and_channel}].  
These noises are either weakly symmetric [W.S.; see Eq.~\eqref{weak_sym_cond_Kraus-2} for the definition] 
or strongly symmetric [S.S.; \eqref{strong_sym_cond_Kraus}] under the symmetry operations $\mathbb{Z}_{2} \times \mathbb{Z}_{2}$ 
and time-reversal. 
``Noise 1'' is called ``dephasing'' in Ref.~\onlinecite{deGroot-T-S-22}.\label{noises_table}}
\begin{ruledtabular}   
\begin{tabular}{lccc}
 & Kraus op. $\{ \mathcal{K}_{\alpha} \}$  &$\mathbb{Z}_{2} \times \mathbb{Z}_{2}$ & time-reversal \\
\hline
Noise~1 & $\{ 1, \be^{i\pi S_{x}}, \be^{i\pi S_{y}}, \be^{i\pi S_{z}} \}$ & S.S.  & S.S. \\
\hline
Noise~2 & $\{ 1, S_{x}S_{y}S_{z} , \text{and perm.} \}$ & S.S. & W.S. \\
\hline
Noise~3 & $\{ 1, S_{x}S_{y} , S_{y}S_{z} , S_{z}S_{x} \}$ & W.S. & S.S.  \\
\hline
Noise~4 & $\{ 1, \be^{i\pi S_{z}} \}$ & S.S. & S.S.   \\
\end{tabular}
\end{ruledtabular}
\end{table}
\end{center}

To be specific, we restrict ourselves to uncorrelated quantum channels of the following form:
\begin{equation}
\begin{split}
& \mathcal{E} =  \mathcal{E}_1  \circ \mathcal{E}_2 \circ \cdots \circ \mathcal{E}_N  \;, 
\\
& \mathcal{E}_i (\rho) = (1-p) \rho + \frac{p}{n} \sum_{\alpha=1}^{n}  \mathcal{K}_{\alpha} \, \rho \, \mathcal{K}^{\dagger}_{\alpha} \;  .
\end{split}
\label{Kraus_and_channel}  
\end{equation}
Note that $\mathcal{E}_{i}$ acts only on the spin-1 Hilbert space at site-$i$. 
With a slight abuse of terminologies, we call the operators  $\{ \mathcal{K}_\alpha \}$ appearing here the Kraus operators as well 
although they do not satisfy the usual completeness relation.  
Below, we consider the AKLT states subject to the four types of quantum noises shown in Table~\ref{noises_table}.  

Let us begin with Noise~1 defined by the Kraus operators $\{ 1, \be^{i\pi S_{x}}, \be^{i\pi S_{y}}, \be^{i\pi S_{z}} \}$ (see Table~\ref{noises_table}).  
Below, all the numerical simulations were done for the AKLT chain with $N=7$ $S=1$ spins (and two spin-$1/2$s at the ends) and for the error rate $p=0.25$. 
The discrete time-evolution of the fidelity $F_{U}$ for the $Z$-rotation gate $U = \be^{-i\theta Z/2}$ is shown in Fig.~\ref{noise1} 
for several values of $\theta$.  
Since Noise 1 satisfies the strong symmetry condition \eqref{strong_sym_cond_Kraus} for a particular representation 
$\{ 1, \be^{i\pi S^x}, \be^{i\pi S^y}, \be^{i\pi S^z}  \}$ of the on-site $\mathbb{Z}_2 \times \mathbb{Z}_2$, the fidelity $F_{U}$ of the identity gate ($\theta=0$) 
stays at unity as guaranteed by Prop.~\ref{SS_cond_and_identity_fidelity}.   
However, the fidelity of the $Z$-rotation gate $\be^{-i\theta Z/2}$ decays for $\theta \neq 0$, 
which means that we cannot ``practically'' implement the generic $Z$-rotation under this noise, that is, 
the noisy AKLT state is no longer a one-qubit universal resource.  

For Noise~2 in Table~\ref{noises_table}, we obtained essentially the same results (see Fig.~\ref{noise2}).  
The only difference between the two cases is that Noise~1 respects the strong time-reversal symmetry, while Noise~2 does not.  
Therefore, we can conclude that the strong symmetry condition 
with respect to time-reversal symmetry does not affect the gate fidelity essentially.  
This is expected from the fact that the gate fidelity is given by a sum of string-like operators associated to the on-site $\mathbb{Z}_2 \times \mathbb{Z}_2$ 
symmetry as we have shown in Sec.~\ref{calculation_of_the_gate_fidelity}.  

As is seen in Fig.~\ref{noise3}, the situation is drastically different for Noise~3 which respects the on-site $\mathbb{Z}_2 \times \mathbb{Z}_2$-symmetry  
only weakly (see Table~\ref{noises_table}).   
Now the gate fidelity $F_{U}$ quickly damped to its lowest possible value $1/4$ for {\em all} $\theta$ (including $\theta=0$).   
This is to be contrasted to the behavior under Noise 4 (see Fig.~\ref{noise4}) which is strongly symmetric under $\mathbb{Z}_2 \times \mathbb{Z}_2$;  
there $F_{U}$ never decays and stays at a constant \eqref{fidelity_result_gs} determined by the system size $N$ as well as $\theta$.    
Since the value $1/4$ is equal to $F_{U}$ for the maximally mixed state, it is clear that under noises that are strongly symmetric 
with respect to time reversal but not to the on-site $\mathbb{Z}_2 \times \mathbb{Z}_2$, the AKLT state totally loses its computational power.  

\begin{figure}[htb]
\begin{center}
\includegraphics[width=\columnwidth,clip]{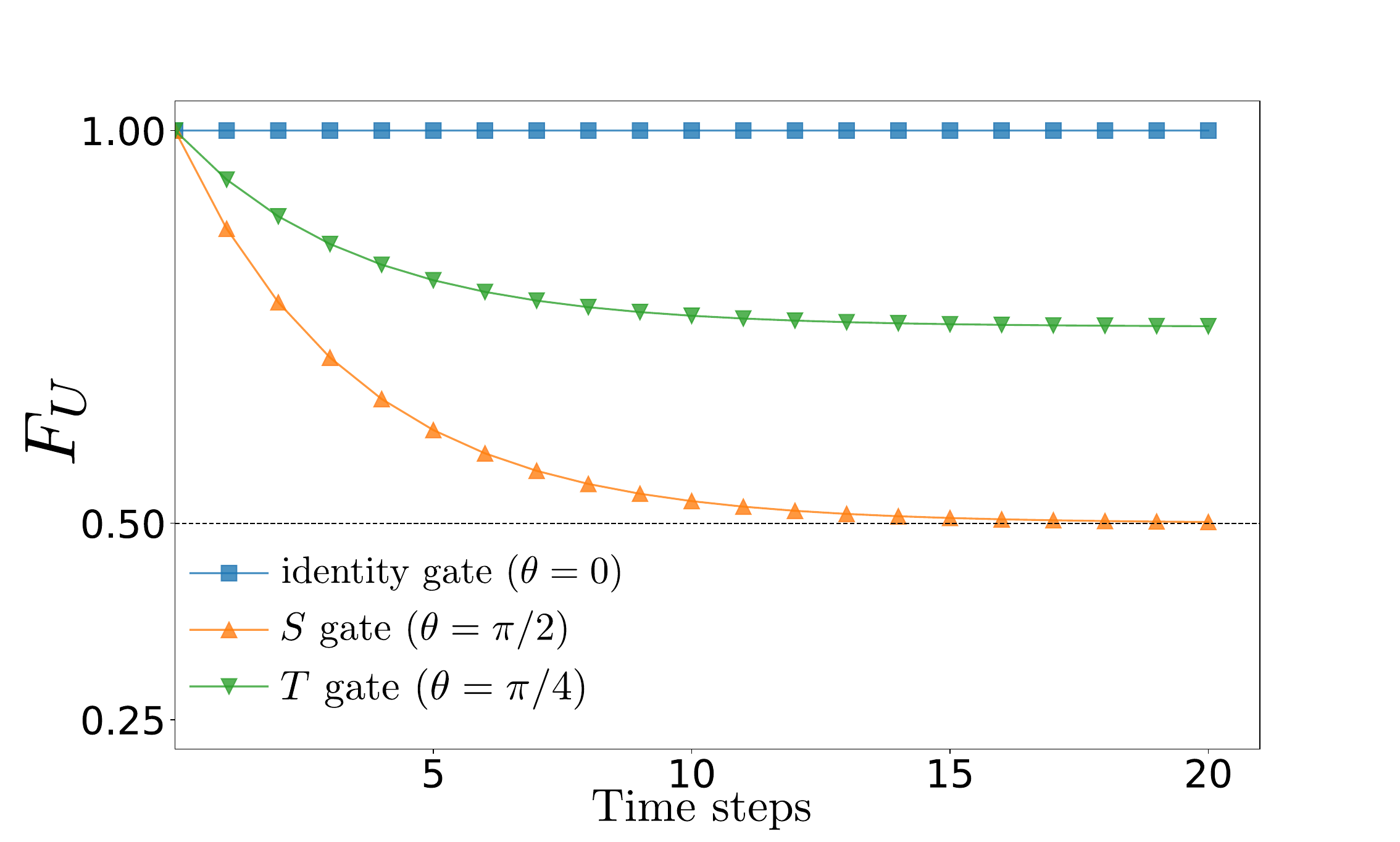}
\caption{%
Discrete time-evolution of the gate fidelity $F_{U}$ for $U = e^{-i\theta Z/2}$ under Noise 1 (called ``dephasing'' in 
Ref.~\onlinecite{deGroot-T-S-22}) 
in Table~\ref{noises_table}. The asymptotic value for the gate $e^{-i\theta Z/2}$ is given by $(1+\text{cos}^2\theta)/2$.  
\label{noise1}}
\end{center}
\end{figure}
\begin{figure}[htb]
\begin{center}
\includegraphics[width=\columnwidth,clip]{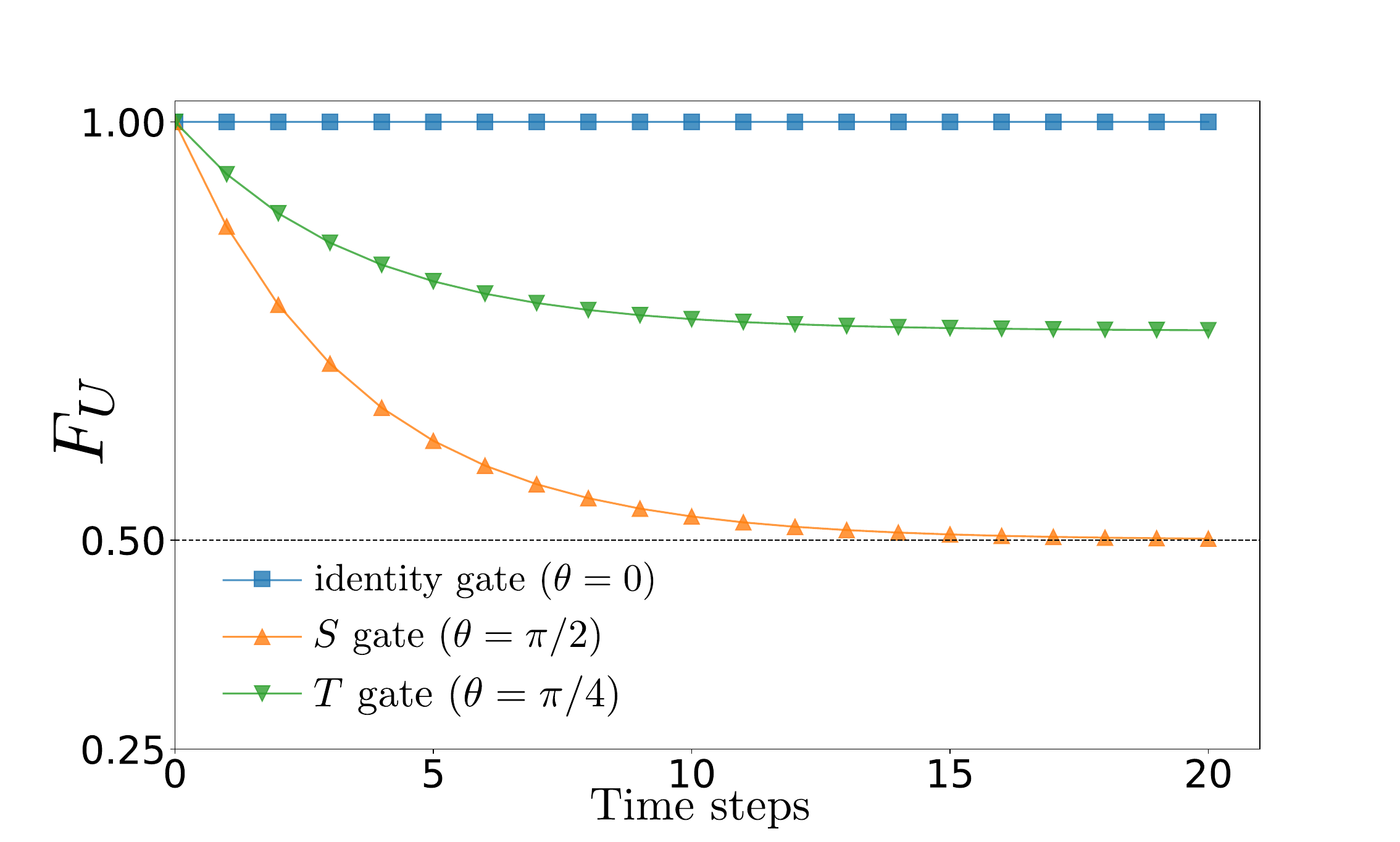}
\caption{%
Discrete time-evolution of the gate fidelity $F_{U}$ for $U = e^{-i\theta Z/2}$ under Noise 2 in Table~\ref{noises_table}.  
The asymptotic values are the same as in Fig.~\ref{noise1}.
\label{noise2}}
\end{center}
\end{figure}
\begin{figure}[htb]
\begin{center}
\includegraphics[width=\columnwidth,clip]{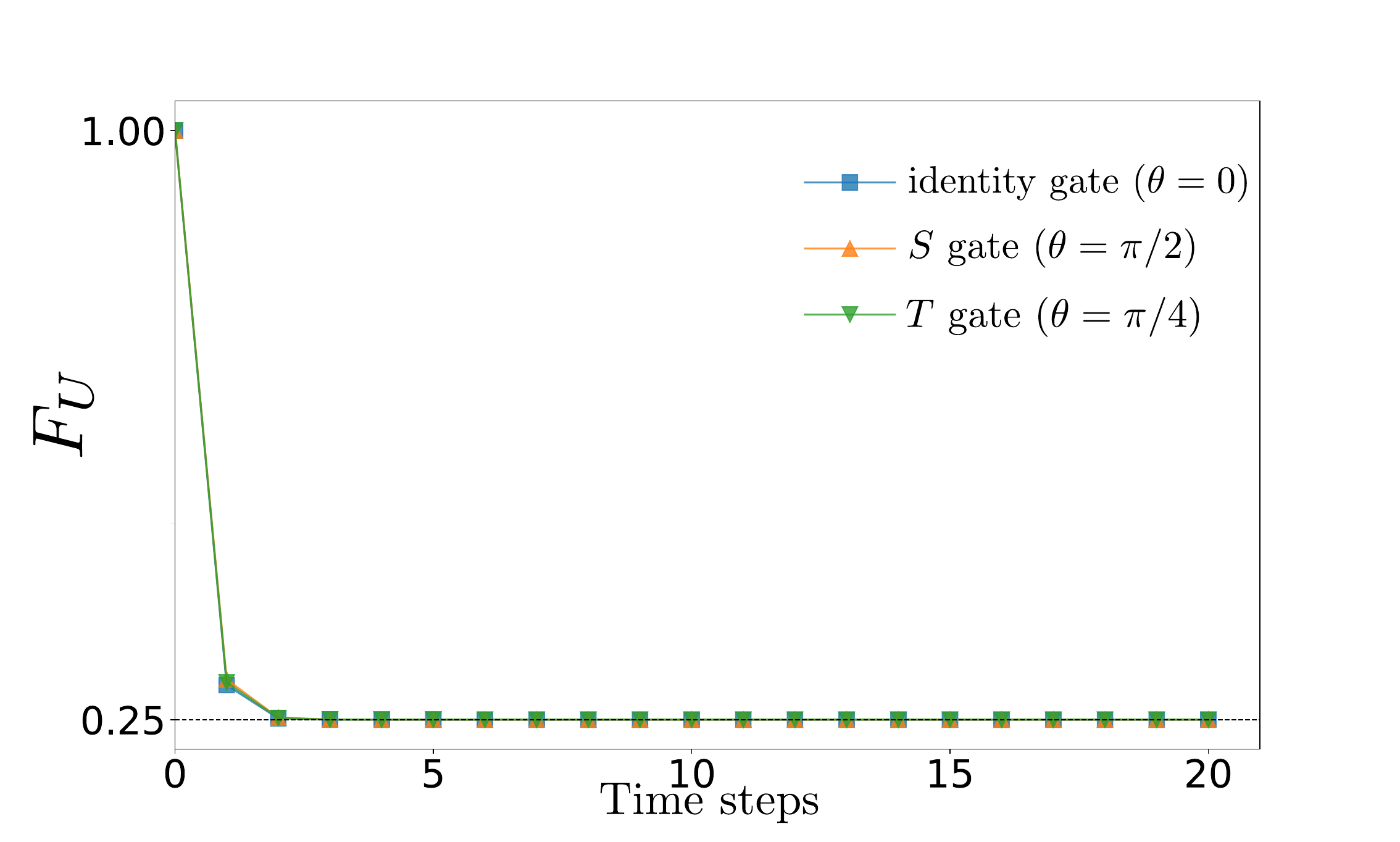}
\caption{%
Discrete time-evolution of the gate fidelity $F_{U}$ for $U = e^{-i\theta Z/2}$ under the weakly symmetric Noise 3 (see Table~\ref{noises_table}).  
For all cases (including $\theta=0$), $F_{U}$ quickly converges to its minimal value 1/4.  
As noted in the main text, this value is the same as in the case of MBQC on the maximally mixed state. 
\label{noise3}}
\end{center}
\end{figure}

The weak symmetry condition \eqref{weak_sym_cond_original} guarantees that if the initial $\rho$ has the $\mathbb{Z}_2\times \mathbb{Z}_2$ symmetry, 
then the decohered $\mathcal{E}(\rho)$ respects the same symmetry: $U_{g} \mathcal{E}[\rho]  U_{g}^{\dagger} 
= \mathcal{E}\left[ U_{g} \rho U_{g}^{\dagger} \right] = \mathcal{E}[\rho]$.    
Therefore, the fact that we are not able to implement the MBQC on the state decohered by a weakly symmetric noise suggests 
the following two possibilities. 
The first is that the preservation of the protecting symmetry ($\mathbb{Z}_2\times \mathbb{Z}_2$, here) alone does not ensure 
the existence of the symmetry-protected edge states 
unlike in the case of the pure state.  In fact, it has been demonstrated recently \cite{Paszko-R-S-P-23} that the strong symmetry condition is necessary 
for the well-defined edge states to exist in the steady states of the Lindbladian dynamics.  
The second possibility is that the protected edge states still exist, whereas we cannot accurately manipulate them.  
Unfortunately, however, we are not able to determine which of the two is the case based on our present numerical results.
\begin{figure}[htb]
\begin{center}
\includegraphics[width=\columnwidth,clip]{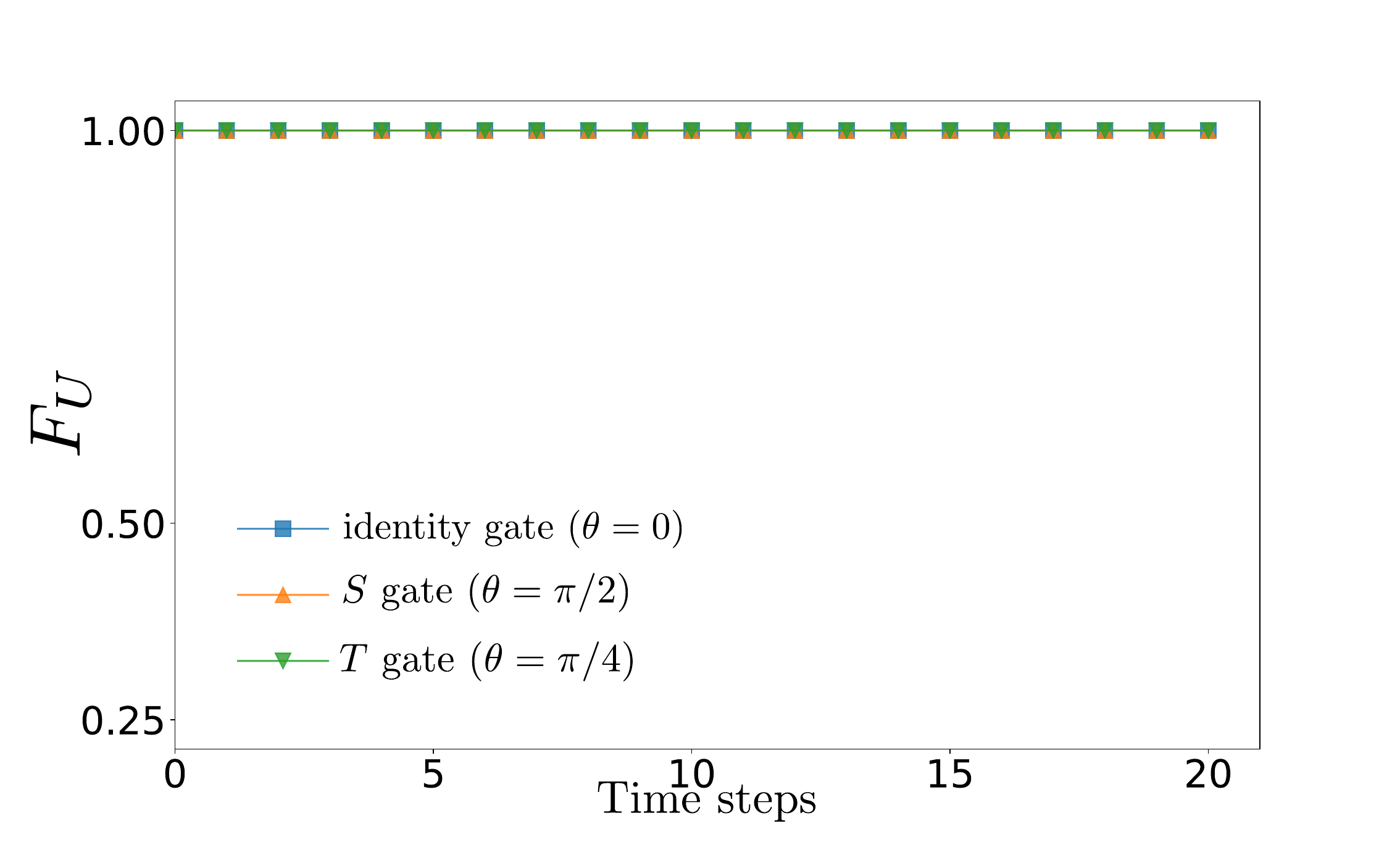}
\caption{%
Discrete time-evolution of the gate fidelity $F_{U}$ for $U = e^{-i\theta Z/2}$ under Noise 4 in Table~\ref{noises_table}.  
The $F_{U}$ stays at a constant $1-\frac{1}{2\cdot 3^N} (1-\text{cos}\theta)$ ($N=7$),  
which is the value for the {\em pure} AKLT state $\rho=\ketbra{ \psi_{\text{AKLT}} }$ [see Eq.~\eqref{fidelity_result_gs}].
\label{noise4}}
\end{center}
\end{figure}

Next, let us compare the results in Fig.~\ref{noise1} (obtained for Noise 1) and in Fig.~\ref{noise4} (for Noise 4).  
Both of types satisfy the strong symmetry conditions for $\mathbb{Z}_2 \times \mathbb{Z}_2$ and time reversal (see Table~\ref{noises_table}).  
Then, what causes the difference between these noises?  
It is another representation $\{ 1, \be^{i\pi \tilde{S}^x(\theta)}, \be^{i\pi \tilde{S}^y(\theta)}, \be^{i\pi S^z}  \}$ of $\mathbb{Z}_2 \times \mathbb{Z}_2$ 
that makes the two cases different.  This may be understood by the following proposition which is proved in Appendix~\ref{sec:proof-Prop2}.  
\begin{prop} \label{condition_for_z_rotation}
We can implement the $Z$-rotation gate $\be^{-i\theta Z/2}$ with the same fidelity as in the case of the pure AKLT state 
if and only if the quantum noise satisfies the strong symmetry condition with respect to the ``rotated'' representation 
$\{ 1, \be^{i\pi \tilde{S}^x(\theta)}, \be^{i\pi \tilde{S}^y(\theta)}, \be^{i\pi S^z}  \}$ of $\mathbb{Z}_2 \times \mathbb{Z}_2$ for all $\theta$.
\end{prop}

Having obtained the condition for the $Z$-rotation gate to be implementable on the decohered states, we now look for  
the condition for one-qubit universal computation.  Specifically, we ask if the identity $1$, the $X$, and $Z$-rotation gates can be executed.  
For the first two gates to be executable, the Kraus operators $\mathcal{K}_\alpha$ in Eq.~\eqref{Kraus_and_channel} must be strongly symmetric 
with respect to both the canonical and rotated representations of $\mathbb{Z}_2 \times \mathbb{Z}_2$, which means 
$\mathcal{K}_\alpha = \text{diag}(a_\alpha,b_\alpha,a_\alpha)$ 
in the $S^{z}$-diagonal basis (with a constraint $\sum_{\alpha} |a_\alpha|^2 =\sum_{\alpha} |b_\alpha|^2 =1$).  
For the fidelity of the $X$-rotation gate $\be^{-i\theta X/2}$ not to decay, $K_\alpha$ should have 
the same form $\text{diag}(a_\alpha,b_\alpha,a_\alpha)$ in the $S^{x}$-diagonal basis, which is possible only when $a_{\alpha}=b_{\alpha}$. 
Therefore, to implement both the $Z$-rotation and the $X$-rotation gates $\be^{-i\theta Z/2}$ and $\be^{-i\theta' X/2}$ with high fidelity, 
$\mathcal{K}_\alpha$ should be proportional to the identity $I_3$ for all $\alpha$, i.e., the quantum channel $\mathcal{E}$ 
in Eq.~\eqref{Kraus_and_channel} is identity. This may be summarized as the following statement:
\begin{prop} \label{absence_universal_computation}
When we consider the discrete time-evolution, we can realize the one-qubit universal computation on $\mathcal{E}(\rho)$ 
if and only if 
the channel is trivial: $\mathcal{E}(\rho)=\rho$. 
\end{prop}
Of course, in the case of continuous time-evolution, the time interval can be chosen arbitrarily and 
we do not rule out the possibility that noisy AKLT states could work as a one-qubit universal resource 
for an interval sufficiently longer than the time scale for experimentally manipulating qubits.
\subsection{Phase structure from computational-power perspective} 
\label{summary_sec_4}
Now, let us summarize what we conclude from the results in this section.  
In Sec.~\ref{sec:symm_cond_fidelity}, we showed that quantum noise respecting the $\mathbb{Z}_2 \times \mathbb{Z}_2$ symmetry 
$\{ 1, \be^{i\pi S^x}, \be^{i\pi S^y}, \be^{i\pi S^z} \}$ in the strong sense is the necessary and sufficient condition 
for implementing the identity gate (Prop.~\ref{SS_cond_and_identity_fidelity}).  
To implement the $Z$-rotation gate $\be^{-i\theta Z/2}$ in our MBQC scheme, we need a much stronger condition 
as has been shown in Sec.~\ref{Z_rotation_fidelity_and_noise}; if  
a quantum noise satisfies the strong symmetry condition for {\em all} the $\mathbb{Z}_2 {\times} \mathbb{Z}_2$ representations of the form 
$\{ 1, \be^{i\pi \tilde{S^x}(\theta)}, \be^{i\pi \tilde{S^y}(\theta)}, \be^{i\pi S^z}  \}$ (not only for $\theta=0$), 
then we can implement $Z$-rotation gate without decay of the fidelity (Prop.~\ref{condition_for_z_rotation}).
In general,  the capability of high-fidelity one-qubit universal computation within our MBQC protocol is gone in any non-trivial quantum channel 
(Prop.~\ref{absence_universal_computation}).  

Our initial goal was to characterize the SPT order in open quantum systems from the perspective of the MBQC computational power.  
The method using the gate fidelity of the MBQC not only yields the results consistent with those of the previous research \cite{deGroot-T-S-22} 
but also suggests a richer structure in the SPT order in open quantum systems.  
Specifically, as is seen in Fig.~\ref{summary}, we can categorize the noisy AKLT states according to what kind of unitary gates we can implement 
by the MBQC on the decohered states.  
Here, we say that a given gate $U$ can be implemented when the corresponding fidelity $F_{U}$ does not decrease from its value 
in the pure case [e.g., Eq.~\eqref{fidelity_result_gs}].  
According to this scheme, the mixed-state SPT phase in the sense of Ref.~\onlinecite{deGroot-T-S-22} corresponds to the decohered AKLT states 
in which the identity gate can be implemented (see the blue region in Fig.~\ref{summary}).  
Among the dissipation considered here, the Noise 1 and 2, that are strongly symmetric with respect to $\mathbb{Z}_2 {\times} \mathbb{Z}_2$ 
defined by $\{ 1, \be^{i\pi S^x}, \be^{i\pi S^y}, \be^{i\pi S^z} \}$, 
lead the initial AKLT state to this subphase.  

In another phase (the green region) that results from
strongly-symmetric noises (e.g., the Noise 4)
satisfying the condition \eqref{strong_sym_cond_Kraus} both for the
canonical and for {\em another}
representation of $\mathbb{Z}_2 {\times} \mathbb{Z}_2$, the
$Z$-rotation gate is also possible.
The non-local order parameter [see Eq.~\eqref{eqn:3rd-4th-tem-in-FU}] that captures this phase is not the standard string order parameter 
used in the previous treatment \cite{deGroot-T-S-22}.  
Phases capable of the universal MBQC (the red point in Fig.~\ref{summary}) result only from a trivial channel.  
This way, the gate fidelity reveals the rich structure in the decohered AKLT states that cannot be captured 
by the string order parameter alone.  

On the other hand, there remain delicate problems when trying to characterize the SPT order in open quantum systems 
based solely on the computational power.
First of all, as we have seen in Sec.~\ref{sec:numerics}, the fidelity-based approach cannot capture the non-triviality of the SPT phases 
protected by symmetries (e.g., time-reversal) other than the on-site $\mathbb{Z}_2\times \mathbb{Z}_2$, although, in the pure case,  
the Haldane phase is protected also by time-reversal symmetry.   

Another immediate problem is related to the MBQC scheme itself.  
As repeatedly noted, we have only considered a simple MBQC protocol which has been originally introduced \cite{Brennen-M-08} for the pure AKLT state.  
Therefore, with more elaborated protocols (e.g., the one proposed in Ref.~\onlinecite{Stephen-W-P-W-R-17}), 
the one-qubit universal computation might be made possible even on the AKLT state subject to non-trivial noises. 
If this is the case, the picture presented in Fig.~\ref{summary} must be modified accordingly.    

To shed some light on the above issue, let us try the following strategy.  
In the case of pure states, the MPS-based approach to the SPT phases has been a great success; most of the important features 
can be understood in terms of the properties of the MPS tensor representing the non-trivial state \cite{Pollmann-T-B-O-10,Schuch-PG-C-11,Chen-G-W-11}. 
Also, it is known that the computational power of {\em any} ground states in the Haldane phase derives from the particular common structure 
of the MPS tensor \cite{Else-S-B-D-12}, 
and the MBQC protocol improved accordingly ensures the one-qubit universal computation on all of them \cite{Stephen-W-P-W-R-17}.   
To apply a similar strategy to mixed states, we first need to represent a given mixed state 
in the form of the matrix-product operator (MPO) \cite{Cirac-G-S-V-RMP-21}, in which the density operator is represented as 
a product of the MPO tensors: \raisebox{-2ex}{\includegraphics[scale=0.2,clip]{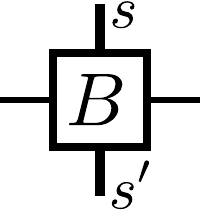}} (with $s$ and $s^{\prime}$ denoting 
the physical states $|s\rangle$ and $\langle s^{\prime}|$, respectively).      
Then, a natural question arises whether it is possible to extract the non-triviality directly from the MPO tensor of the mixed state.  
In this respect, we can rephrase Prop.~\ref{SS_cond_and_identity_fidelity} in the language of the MPO tensor as follows:
\begin{prop}\label{MPOtensor_and_computaitonal_power}
Let $\ket{\psi}$ be a ground state in the Haldane phase and $\mathcal{E}$ be an uncorrelated noise.  
Then, the diagonal elements of the MPO tensors representing the states $\ketbra{\psi}{\psi}$ and $\mathcal{E}(\ketbra{\psi}{\psi})$ in the basis of the zero-eigenvectors of $ S^x, S^y,$ and $S^z$
are the same, i.e., share the same computational power for the identity gate, 
if and only if $\mathcal{E}$ is strongly symmetric with respect to the canonical representation 
$\{ 1, \be^{i\pi S^x}, \be^{i\pi S^y}, \be^{i\pi S^z} \}$ of the $\mathbb{Z}_2\times \mathbb{Z}_2$-symmetry. 
\end{prop}
The proof is given in Appendix~\ref{sec:MPO_tensor_product_decomposition}.   
This guarantees that if a quantum channel satisfies the strong symmetry condition (for the canonical representation 
of $\mathbb{Z}_2\times \mathbb{Z}_2$), the identity gate can be implemented with fidelity $1$ on {\em any} decohered states resulting from 
the ground states in the Haldane phase.   
The scope of this Prop.~\ref{MPOtensor_and_computaitonal_power} is limited to the mixed states 
originating from pure SPT ground states, whereas we believe that Prop.~\ref{MPOtensor_and_computaitonal_power}, 
together with other results obtained here, serves as a stepping stone in understanding interacting many-body phases in open quantum systems 
from the quantum-computational perspectives.  Understanding the relation between the MPO-based picture presented here 
and the approaches based on the (super)MPS in the doubled Hilbert space \cite{Nieuwenburg-H-14,Verissimo-L-O-23,Ma-T-24} 
is an important open question.  

\begin{figure}[tbh]
\begin{center}
\includegraphics[width=\columnwidth,clip]{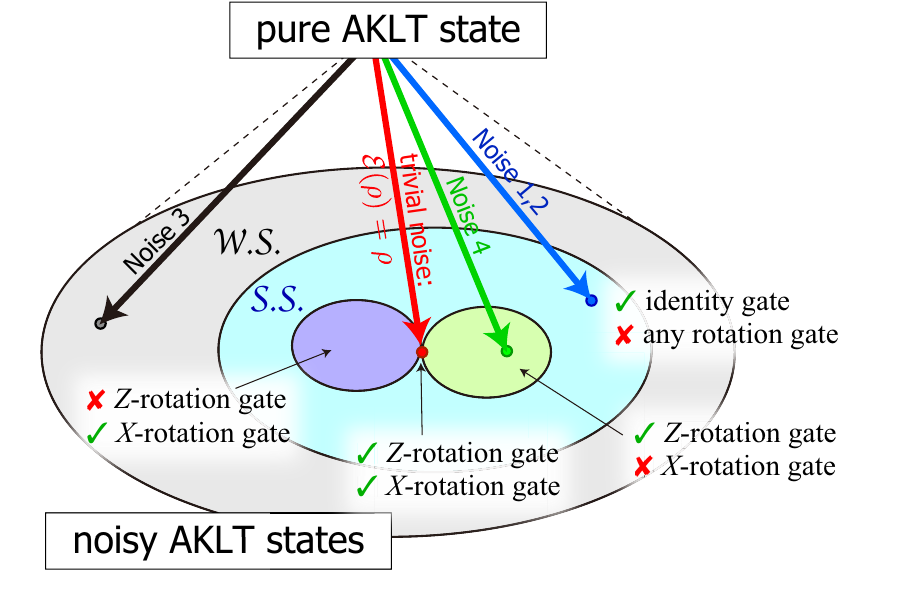}
\caption{%
Discrete time-evolution of the pure SPT (i.e., AKLT) state by the four types of quantum noises (see Table~\ref{noises_table}). 
The ``S.S.'' and ``W.S.'' respectively mean that the corresponding noises are strongly and weakly symmetric, respectively.  
Even with strongly-symmetric noises, the noisy SPT (AKLT) states are further divided into several subcategories 
according to the ability of carrying out quantum gates. If we require that both the $X$ and $Z$-rotation gates should be executable 
with high fidelity, only a trivial noise $\mathcal{E}(\rho)=\rho$ is allowed.  
\label{summary}}
\end{center}
\end{figure}

\section{Conclusion and Outlook}
\label{sec:conclusion}

The gate fidelity, which is related to the stability of the edge states, is a quantitative measure in quantum information science 
that assesses how accurately we can execute the MBQC by quantifying the ability to manipulate the qubit.  
In this paper, we proposed to use the gate fidelity of the MBQC as a detector of the SPT signature in generic mixed states.  
We observed that, depending on the unitary gates under consideration, the gate fidelity consists of different string-like operators 
which are associated with the on-site $\mathbb{Z}_2 \times \mathbb{Z}_2$ symmetry.  

First, we defined the SPT order in open quantum system by the capability of executing the identity gate with the unit fidelity, 
thereby obtaining results consistent with those of the existing approaches.  
In particular, we established (Prop.~\ref{SS_cond_and_identity_fidelity}) that the high-fidelity realization of the identity gate in the decohered state 
is equivalent to the (uncorrelated) quantum noise respecting the $\mathbb{Z}_2 \times \mathbb{Z}_2$ symmetry in the strong sense. 
This is the translation of the symmetry condition for quantum channels that preserve the SPT order into the language 
of the MBQC computational power.  
To put it another way, a quantum noise respecting the protecting (on-site) symmetry in the strong sense preserves not only 
the corresponding string order parameters but also the computational power of the initial SPT phase.  

We then investigated the condition for the quantum noises under which we can perform the $Z$-rotation gate 
$U=U_{Z}(\theta)=\be^{-i\theta Z/2}$ on noisy AKLT states 
and found that, to implement $U_{Z}(\theta \neq 0)$ with high fidelity, it is not sufficient that 
the quantum noise satisfies the strong symmetry condition only for a particular representation of $\mathbb{Z}_2 \times \mathbb{Z}_2$ 
(see Prop.~\ref{condition_for_z_rotation}).  
Moreover, to realize one-qubit universal computation on the decohered state, a much stronger condition is required.  
In fact, we showed that, as far as we stick to the simplest MBQC protocol used in this paper, 
no non-trivial quantum noise exists that allows the implementation of one-qubit universal computation 
on the noisy AKLT states (see Prop.~\ref{absence_universal_computation})   
\footnote{%
We do not exclude here the possibility that MBQC on noisy AKLT states could potentially be one-qubit universal at the practical level {\em for a finite duration} even in this simple MBQC scheme. To investigate what kind of quantum channels leave the noisy AKLT states one-qubit universal would be one of the most practically important problems 
from the viewpoint of both condensed matter physics and quantum information.  }.  
These may suggest that, as shown in Fig.~\ref{summary}, 
the decohered ``SPT'' phase defined by the perfect performance of the identity gate (or equivalently, by the persisting string 
order parameters) is further decomposed into subphases according to 
the capability of implementing other gates $U$ with reasonably high fidelity.

Our MBQC-based approach to the SPT order in open quantum systems using the gate fidelity yielded a picture generally consistent with those from other order-parameter-based approaches \cite{deGroot-T-S-22,Paszko-R-S-P-23}. Nevertheless, there still remain a few important questions. 
First of all, as shown in Secs.~\ref{sec:results} and \ref{sec:effect-noise}, the gate fidelity is always related to (generalized) string-order parameters 
associated with the on-site symmetry and the behavior of the fidelity has little to do with the presence or absence of the time-reversal symmetry. 
Therefore, for the moment, we do not know how to formulate the mixed-state SPT phases associated to the time-reversal symmetry 
by the computational power.    
Another problem is related to the choice of the MBQC protocol.   
Although Prop.~\ref{absence_universal_computation} forbids the one-qubit universal computation in the presence of non-trivial noises, 
we cannot exclude the possibility that improved MBQC protocols will change the conclusion of Prop.~\ref{absence_universal_computation} 
thereby modifying the global picture presented in Fig.~\ref{summary}.   
Finally, the relation between our formulation and other ones (e.g., the ones based on the doubled Hilbert space 
in Refs.~\onlinecite{Nieuwenburg-H-14,Verissimo-L-O-23,Ma-T-24}) 
is not fully understood, though our definition of the SPT order in open quantum systems leads to conclusions 
consistent with those in, e.g., Ref.~\onlinecite{deGroot-T-S-22}.

A promising way to answer these questions (especially, the latter two) may be to directly extract the information on the computational power 
from the MPO tensor of the mixed state in a way parallel to what we did to understand the uniform computational power 
of the SPT phases \cite{Else-S-B-D-12,Stephen-W-P-W-R-17}.  
In this respect, Prop.~\ref{MPOtensor_and_computaitonal_power} and its Corollary~\ref{identity_gate_and_MPO} 
may be the first step towards the exploration into these interesting open problems.  

There are experimental realizations of the AKLT state and its application for the quantum teleportation on superconduting qubits \cite{Smith-C-W-G-23} and the MBQC on photonic AKLT state \cite{Kaltenbaek-L-Z-B-R-10} with high fidelity in both experiments.  
Also, the Haldane phase has been realized \cite{Sompet-et-al-SPT-Hubbard-22} in a ultra-cold Fermi gas in which the real-space measurement 
of non-local string correlations is possible.  
We believe that our approach to formulate the interacting SPT order in open quantum systems from the viewpoint 
of the computational power holds significant importance not only in the condensed-matter-physics contexts 
but also in practical applications, especially in the rapidly advancing field of quantum technology.

\section*{Acknowledgements}
The authors would like to thank Y.~Nakata and A.~Turzillo for helpful discussions and K.~Fujii for suggesting us the use of the gate fidelity.
One of the authors (R.M.) is supported by JST, the establishment of university fellowships towards the creation of science technology innovation, Grant Number JPMJFS2123. 
The author (K.T.) is supported in part by JSPS KAKENHI Grant No. 18K03455 and No. 21K03401.  

\appendix
\section{Derivation of gate fidelity $F_{U_{Z}(\theta)}$}
\label{sec:deriv-FU}
\subsection{$ZZ$-term}
The calculation of the $Z_{\text{in}} Z_{\text{out}}$-term starts from Eq.~\eqref{calculation_of_ZZ_term_1}:
\begin{equation}
\begin{split}
&\displaystyle \text{Tr}_{\text{in, out}}\left[\rho_{U_Z(\theta)} Z_{\text{in}} Z_{\text{out}} \right]    \\ 
&= 
\widetilde{\text{Tr}} \sum_{\boldsymbol{m}} \left[ B^{(\boldsymbol{m})}_{\text{out}} \boldsymbol{P}(\boldsymbol{m}) 
\tilde{\rho} \, \boldsymbol{P}(\boldsymbol{m}) B^{(\boldsymbol{m}) \dagger }_{\text{out}}  Z_{\text{in}} Z_{\text{out}} \right]    \\
&=
\widetilde{\text{Tr}} \sum_{\boldsymbol{m}} \left[ \tilde{\rho} \, Z_{\text{in}}  \boldsymbol{P}(\boldsymbol{m}) B^{(\boldsymbol{m}) \dagger }_{\text{out}}  Z_{\text{out}}  B^{(\boldsymbol{m})}_{\text{out}} \boldsymbol{P}(\boldsymbol{m})  \right] \; ,
\end{split}
\label{calculation_of_ZZ_term_1b}
\end{equation}
where $\widetilde{\text{Tr}}$ denotes the trace over the entire system consisting of the bulk spin-1 chain and the two spin-1/2s at the edges. 
Without loss of generality, we can postulate the by-product operator $B^{(\boldsymbol{m})}_{\text{out}}$ as:
\begin{equation}
B^{(\boldsymbol{m})}_{\text{out}} = X_{\text{out}}^{r^{(\bolm)}_X+1}Z_{\text{out}}^{r^{(\bolm)}_Z+1}  \; ,
\label{eqn:by-product-op-postulated}
\end{equation} 
where $r^{(\bolm)}_X$ is the number of times we get one of the results $\{ \ket{\alpha_{\theta}}, \ket{\beta_{\theta}}, \ket{\alpha_{0}}, \ket{\beta_{0}} \}$ 
in successive measurements, and $r^{(\bolm)}_Z$ is that for $\{ \ket{\beta_{\theta}}, \ket{\beta_{0}}, \ket{\gamma} \}$.  
For later convenience, we also introduce $r^{(\bolm)}_Y$ which counts how many times we obtain one of 
$\{ \ket{\alpha_{\theta}}, \ket{\alpha_{0}}, \ket{\gamma} \}$.  
By definition, these three integers sum to an even integer, i.e., 
\begin{equation}
(-1)^{r^{(\bolm)}_X+r^{(\bolm)}_Y+r^{(\bolm)}_Z} = 1  \; .
\label{eqn:rX-rY-rZ-constraint}
\end{equation}

Then, we can move one of the by-product operators to annihilate them leaving a sign factor $(-1)^{r^{(\bolm)}_X+1}$:
\begin{equation}
B^{(\boldsymbol{m}) \dagger }_{\text{out}}  Z_{\text{out}}  B^{(\boldsymbol{m})}_{\text{out}} \boldsymbol{P}(\boldsymbol{m}) 
= (-1)^{r^{(\bolm)}_X+1} Z_{\text{out}} \boldsymbol{P}(\boldsymbol{m})   \; .
\end{equation}
To rewrite the sign factor $(-1)^{r^{(\bolm)}_{X}}$ in terms of the string operator, we note that the states 
$\{ \ket{\alpha_{\theta}},\ket{\beta_{\theta}},\ket{\gamma} \}$ are the eigenstates of $\be^{i \pi S^{z}}$ for {\em any} $\theta$ 
(including $\theta=0$). 
Specifically, we we use the following relations:
\begin{subequations}
\begin{align}
\begin{split}
& \mathcal{P}_j (m_j) \be^{i \pi S_j^z} = 
\begin{cases}
  - \ketbra{\alpha_{\theta}}{\alpha_{\theta}}_{j} &  m_{j} =\alpha_{\theta} \\
  - \ketbra{\beta_{\theta}}{\beta_{\theta}}_{j} & m_{j} =\beta_{\theta}\\
  \ket{\gamma}\bra{\gamma}_{j}  & m_{j} =\gamma
\end{cases}
 \\
& = \be^{i \pi S_j^z} \mathcal{P}_j (m_j)  \;, 
\end{split}
\\
\begin{split}
& \mathcal{P}'_j (m_j) \be^{i \pi S_j^z} =
\begin{cases}
  - \ketbra{\alpha_{\theta=0}}{\alpha_{\theta=0}}_{j} &  m_{j} =\alpha_{\theta=0} \\
  - \ketbra{\beta_{\theta=0}}{\beta_{\theta=0}}_{j} & m_{j} =\beta_{\theta=0}\\
  \ket{\gamma}\bra{\gamma}_{j}  & m_{j} =\gamma
\end{cases} \\
& = \be^{i \pi S_j^z} \mathcal{P}'_j (m_j)  \; .
\end{split}
\end{align} 
\label{eigen_eq_3}
\end{subequations}
From Eqs.~\eqref{eigen_eq_3}, we immediately see that the sign factor originating from $X_{\text{out}}^{r^{(\bolm)}_X+1}$ in  
the by-product operator is now replaced with the action of the string operator:
\begin{equation}
(-1)^{r^{(\bolm)}_X+1} Z_{\text{out}} \boldsymbol{\mathcal{P}}(\boldsymbol{m}) 
= - \boldsymbol{\mathcal{P}}(\boldsymbol{m}) \left( \prod_{j=1}^{N} \be^{i\pi S_j^z}  \right) Z_{\text{out}} \; . 
\label{eqn:sign-to-string}
\end{equation}
As the $\boldsymbol{m}$-dependence appears only in $\boldsymbol{\mathcal{P}}(\boldsymbol{m})$, we can explicitly carry out 
the summation over all the possible measurement outcomes $\{ \boldsymbol{m} \}$ to obtain the final result \eqref{fidelity_ZZ_term}:
\begin{equation}
\begin{split}
&\displaystyle \text{Tr}_{\text{in, out}}\left[\rho_{U_Z(\theta)} Z_{\text{in}} Z_{\text{out}} \right]   \\
&= 
- \widetilde{\text{Tr}} \sum_{\boldsymbol{m}} \left[ \tilde{\rho} \,  Z_{\text{in}}  \bolP(\boldsymbol{m})^{2} 
\left( \prod_{j=1}^{N} \be^{i\pi S_j^z}  \right) Z_{\text{out}}  \right]  \\
&= - \widetilde{\text{Tr}}  \left[ \tilde{\rho} \,  Z_{\text{in}}    
\left( \sum_{\boldsymbol{m}}   
\boldsymbol{P}(\boldsymbol{m}) \right) \left( \prod_{j=1}^{N}  \be^{i\pi S_j^z} \right)  Z_{\text{out}}  \right]   \\
&= - \widetilde{\text{Tr}}  \left[ 
\tilde{\rho} \,  Z_{\text{in}}  \left( \prod_{j=1}^{N}   \be^{i\pi S_j^z} \right)  Z_{\text{out}}  \right]   \; ,
\end{split}
\end{equation}
where we have utilized:
\begin{equation}
\begin{split}
& \sum_{\boldsymbol{m}} \bolP(\bolm) \\
& = \sum_{\boldsymbol{m}} P'_{N}(m_N)\cdots P'_{l(\boldsymbol{m})+1}(m_{l(\boldsymbol{m}) +1})   \\
& \qquad \quad   \times P_{l(\boldsymbol{m})}(m_{l(\boldsymbol{m}) }) \cdots P_{1}(m_1)  = I_3 \; .
\end{split}
\end{equation}

\subsection{The other terms}
The third and fourth terms are calculated as:
\begin{subequations}
\begin{equation}
\begin{split}
&\displaystyle{\text{Tr}_{\text{in, out}}}  \left[\rho_{U_Z(\theta)} X_{\text{in}} \, \be^{- i \theta Z_{\text{out}} /2} X_{\text{out}} \, 
\be^{i \theta Z_{\text{out}} /2} \right]  \\ 
&= \cos \theta\, 
\widetilde{\text{Tr}} \sum_{\boldsymbol{m}} \left[ \tilde{\rho} X_{\text{in}}  \boldsymbol{\mathcal{P}} (\boldsymbol{m}) 
B^{(\boldsymbol{m}) \dagger }_{\text{out}}  X_{\text{out}}  B^{(\boldsymbol{m})}_{\text{out}} 
\boldsymbol{\mathcal{P}} (\boldsymbol{m})  \right]  \\
& \phantom{=} 
+ \sin \theta\, 
\widetilde{\text{Tr}} \sum_{\boldsymbol{m}} \left[ \tilde{\rho} X_{\text{in}}  \boldsymbol{\mathcal{P}} (\boldsymbol{m}) 
B^{(\boldsymbol{m}) \dagger }_{\text{out}}  
Y_{\text{out}}  B^{(\boldsymbol{m})}_{\text{out}} \boldsymbol{\mathcal{P}} (\boldsymbol{m})  \right]  
\end{split}
\label{eqn:FU-3rd-term}
\end{equation}
and 
\begin{equation}
\begin{split}
&\displaystyle{\text{Tr}_{\text{in, out}}}  \left[\rho_{U_Z(\theta)} 
X_{\text{in}}Z_{\text{in}} \, \be^{-i\theta Z_{\text{out}}/2} X_{\text{out}} Z_{\text{out}} \, \be^{i\theta Z_{\text{out}}/2}   \right] \\ 
&=-  \displaystyle{\text{Tr}_{\text{in, out}}}  \left[\rho_{U_Z(\theta)} Y_{\text{in}} \, 
\be^{- i \theta Z_{\text{out}} /2} Y_{\text{out}} \, \be^{i \theta Z_{\text{out}} /2} \right]  \\
&=  \sin \theta\, 
\widetilde{\text{Tr}} \sum_{\boldsymbol{m}} \left[ \tilde{\rho} \, Y_{\text{in}}  \boldsymbol{\mathcal{P}} (\boldsymbol{m}) 
B^{(\boldsymbol{m}) \dagger }_{\text{out}}  
X_{\text{out}}  B^{(\boldsymbol{m})}_{\text{out}} \boldsymbol{\mathcal{P}} (\boldsymbol{m})  \right] \\
& \phantom{=} 
- \cos \theta\, 
\widetilde{\text{Tr}} \sum_{\boldsymbol{m}} \left[ \tilde{\rho}\, Y_{\text{in}}  \boldsymbol{\mathcal{P}} (\boldsymbol{m}) 
B^{(\boldsymbol{m}) \dagger }_{\text{out}}  Y_{\text{out}}  B^{(\boldsymbol{m})}_{\text{out}} 
\boldsymbol{\mathcal{P}} (\boldsymbol{m})  \right]   \;. 
\end{split}
\label{eqn:FU-4th-term}
\end{equation}
\end{subequations}
As before, we move one of the by-product operators $B^{(\boldsymbol{m})}_{\text{out}}$ 
and $B^{(\boldsymbol{m}) \dagger }_{\text{out}}$ [see Eq.~\eqref{eqn:by-product-op-postulated}] to make them pair-annihilate: 
\begin{equation}
\begin{split}
& B^{(\boldsymbol{m}) \dagger }_{\text{out}}  X_{\text{out}}  B^{(\boldsymbol{m})}_{\text{out}}
= (-1)^{r_{Z}^{(\bolm)}+1} X_{\text{out}}   \\
& B^{(\boldsymbol{m}) \dagger }_{\text{out}}  Y_{\text{out}}  B^{(\boldsymbol{m})}_{\text{out}}
= (-1)^{r_{X}^{(\bolm)}+r_{Z}^{(\bolm)}} Y_{\text{out}}  
= (-1)^{r_{Y}^{(\bolm)}} Y_{\text{out}}     \; ,
\end{split}
\end{equation}
where we have used Eq.~\eqref{eqn:rX-rY-rZ-constraint}.   
Now we use the same trick as in \eqref{eqn:sign-to-string}.  Specifically, we use the followings to rewrite the sign factors appearing 
when the by-product operator and the Pauli operators are exchanged:
\begin{subequations}
\begin{align}
& \mathcal{P}_j (m_j) \be^{i \pi \tilde{S}_i^x(\theta)} = 
\begin{cases}
   \ketbra{\alpha_{\theta}}{\alpha_{\theta}}_{j} &  m_{j} =\alpha_{\theta} \\
  - \ketbra{\beta_{\theta}}{\beta_{\theta}}_{j} & m_{j} =\beta_{\theta}\\
  - \ket{\gamma}\bra{\gamma}_{j}  & m_{j} =\gamma
\end{cases}
\\
& \mathcal{P}'_j (m_j) \be^{i \pi S_j^x} =
\begin{cases}
   \ketbra{\alpha_{\theta=0}}{\alpha_{\theta=0}}_{j} &  m_{j} =\alpha_{\theta=0} \\
  - \ketbra{\beta_{\theta=0}}{\beta_{\theta=0}}_{j} & m_{j} =\beta_{\theta=0}\\
  - \ket{\gamma}\bra{\gamma}_{j}  & m_{j} =\gamma    \; .
\end{cases} 
\end{align} 
\label{eigen_eq_1}
\end{subequations}
and similar relations for $\tilde{S}_i^y(\theta)$ and $S_{i}^{y}$: 
\begin{subequations}
\begin{align}
& \mathcal{P}_j (m_j) \be^{i \pi \tilde{S}_i^y(\theta)} = 
\begin{cases}
  - \ketbra{\alpha_{\theta}}{\alpha_{\theta}}_{j} &  m_{j} =\alpha_{\theta} \\
   \ketbra{\beta_{\theta}}{\beta_{\theta}}_{j} & m_{j} =\beta_{\theta}\\
  - \ket{\gamma}\bra{\gamma}_{j}  & m_{j} =\gamma
\end{cases}
\\
& \mathcal{P}'_j (m_j) \be^{i \pi S_j^y} =
\begin{cases}
  - \ketbra{\alpha_{\theta=0}}{\alpha_{\theta=0}}_{j} &  m_{j} =\alpha_{\theta=0} \\
   \ketbra{\beta_{\theta=0}}{\beta_{\theta=0}}_{j} & m_{j} =\beta_{\theta=0}\\
  - \ket{\gamma}\bra{\gamma}_{j}  & m_{j} =\gamma    \; .
\end{cases} 
\end{align} 
\label{eigen_eq_2}
\end{subequations}
Note that $\ket{\alpha_{\theta} }$ and $\ket{\beta_{\theta} }$ are no longer the eigenstates of $\be^{i \pi S_i^{x,y}}$.  
Therefore, in contrast to the case of the $Z_{\text{in}} Z_{\text{out}}$-term, we must use two different operators 
$\be^{i \pi \tilde{S}_i^{x,y}(\theta)}$ (twisted) and $\be^{i \pi S_i^{x,y}}$ (untwisted) respectively 
for $\mathcal{P}_j$ and $\mathcal{P}^{\prime}_j$ in order to have the correct sign factors [see Eq.~\eqref{eigen_eq_3}].  
To be specific, we have the following expressions:
\begin{subequations}
\begin{equation}
\begin{split}
& B^{(\boldsymbol{m}) \dagger }_{\text{out}}  X_{\text{out}}  B^{(\boldsymbol{m})}_{\text{out}}
= (-1)^{r^{(\bolm)}_Z+1} X_{\text{out}} \boldsymbol{\mathcal{P}}(\boldsymbol{m}) \\
& = - \bolP (\boldsymbol{m})
\mathcal{O}^{(x)}_{\text{str}}(\theta; l(\bolm))
 X_{\text{out}} \; 
\end{split}
\label{eqn:sign-to-string-2} 
\end{equation}
and
\begin{equation}
\begin{split}
& B^{(\boldsymbol{m}) \dagger }_{\text{out}}  Y_{\text{out}}  B^{(\boldsymbol{m})}_{\text{out}}
= (-1)^{r^{(\bolm)}_Y} Y_{\text{out}} \boldsymbol{\mathcal{P}}(\boldsymbol{m}) \\
& = \bolP (\boldsymbol{m})
\mathcal{O}^{(y)}_{\text{str}}(\theta; l(\bolm))
 Y_{\text{out}} \; ,
\end{split}
\label{eqn:sign-to-string-3} 
\end{equation}
where we have introduced the following generalized string operator:
\begin{equation}
\mathcal{O}^{(a)}_{\text{str}}(\theta; l) := 
\left( \prod_{j=1}^{l} \be^{i \pi \tilde{S}_j^a (\theta)} \right) \left(\prod_{k=l +1}^{N} \be^{i \pi S_k^a } \right)  
\quad (a=x,y)  \; .
\end{equation}
\end{subequations}
Note that the string operator is switched from the twisted one to the untwisted one at the site $l(\bolm)$.  
Plugging Eqs.~\eqref{eqn:sign-to-string-2} and \eqref{eqn:sign-to-string-3} into \eqref{eqn:FU-3rd-term} and \eqref{eqn:FU-4th-term},  
we obtain:
\begin{subequations}
\begin{equation}
\begin{split}
&\displaystyle{\text{Tr}_{\text{in, out}}}  \left[\rho_{U_Z(\theta)} X_{\text{in}} \,  
\be^{-i \theta Z_{\text{out}} /2} X_{\text{out}} \, \be^{i \theta Z_{\text{out}} /2} \right]  \\
&= - \widetilde{\text{Tr}} \sum_{\boldsymbol{m}} \left[ \tilde{\rho} \, X_{\text{in}}  \bolP(\bolm) 
\mathcal{O}^{(x)}_{\text{str}}(\theta; l(\bolm) ) 
X_{\text{out}}  \right] \text{cos} \theta  \\
& \phantom{=}\;\;  
+ \widetilde{\text{Tr}} \sum_{\boldsymbol{m}} \left[ \tilde{\rho} \, X_{\text{in}}  \boldsymbol{P}(\boldsymbol{m}) 
\mathcal{O}^{(y)}_{\text{str}}(\theta; l(\bolm))
Y_{\text{out}}  \right] \text{sin} \theta \label{XX_term}
\end{split}
\end{equation}
and 
\begin{equation}
\begin{split}
&\text{Tr}_{\text{in, out}} \left[\rho_{U_Z(\theta)} X_{\text{in}} Z_{\text{in}} \, 
\be^{-i \theta Z_{\text{out}} /2} Z_{\text{out}} X_{\text{out}} \, \be^{i \theta Z_{\text{out}} /2} \right]   \\
&= - \widetilde{\text{Tr}} \sum_{\boldsymbol{m}} \left[ \tilde{\rho} \,  Y_{\text{in}}  \boldsymbol{P}(\boldsymbol{m}) 
\mathcal{O}^{(x)}_{\text{str}}(\theta; l(\bolm))  
X_{\text{out}}  \right] \text{sin} \theta   \\
& \phantom{=} \;\; - \widetilde{\text{Tr}} \sum_{\boldsymbol{m}} \left[ \tilde{\rho} \, Y_{\text{in}}  \boldsymbol{P}(\boldsymbol{m}) 
\mathcal{O}^{(y)}_{\text{str}}(\theta; l(\bolm))  
Y_{\text{out}}  \right] \text{cos} \theta \label{XZ_term} \;  .
\end{split}
\end{equation}
\end{subequations}
Combining these two reproduces Eq.~\eqref{eqn:3rd-4th-tem-in-FU}.   

\section{Proof of Prop.~\ref{condition_for_z_rotation}}
\label{sec:proof-Prop2}
In order to prove Prop.~\ref{condition_for_z_rotation}, we firstly introduce the following lemma:
\begin{lem}\label{lemma}
A quantum channel described by $\{ K_\alpha \}$ is strongly symmetric for the representation 
$\mathbb{Z}_2 \times \mathbb{Z}_2 = \{ 1, \be^{i\pi \tilde{S}^x(\theta)}, \be^{i\pi \tilde{S}^y(\theta)}, \be^{i\pi S^z}  \}$ if and only if 
\begin{equation}
\begin{split}
& \sum_{p} K_{p}^{\dagger} \mathcal{P}(m_{i} ) K_{p} =  \mathcal{P}(m_{i} )   \\
& \raisebox{-7.5ex}{\includegraphics[scale=0.4,clip]{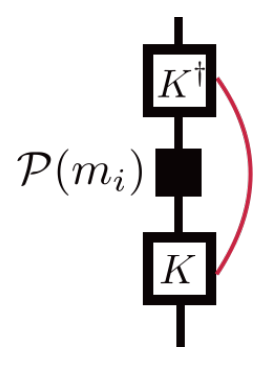}} = \raisebox{-3ex}{\includegraphics[scale=0.4,clip]{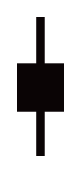}}
 ~~~(m_i = \alpha_\theta, \beta_\theta, \gamma)
 \end{split}
 \label{projection_SS_condition_1}
\end{equation}
holds. 
\end{lem}

To prove this, we start from the following Lemma:
\begin{lem}\label{lemma2}
If $G=\mathbb{Z}_2 \times \mathbb{Z}_2$, the phase of the right-hand side of Eq.\;\eqref{strong_sym_cond_Kraus} 
is $\be^{i \phi(g)}=1$ for all $g\in G$.
\end{lem}
\begin{proof}
Without loss of generality, we can choose the linear representation 
$\{ 1, \be^{i\pi S^x}, \be^{i\pi S^y}, \be^{i\pi S^z} \}$ of $\mathbb{Z}_2 \times \mathbb{Z}_2$.   
By Eq.~\eqref{strong_sym_cond_Kraus}, the representation $U_{g}$ ($U_{g}^{2}=1$) satisfies 
\eq{
 K_{p}
= U_g (U_g K_{p} U_g) U_g = e^{i \phi(g)} U_g K_{p}U_g^{\dagger} = \left( \be^{i \phi (g)}\right)^2 K_{p},
}
for all $p$ since $U_g = U^\dagger_g$.   Therefore, 
\eq{
\be^{i \phi (g)} = \pm 1~~(\text{for } ^\forall g)\;.
}
Next, we show $\be^{i \phi(g)}$ cannot be $-1$ for any $g$.  
To this end, we move to the basis $\{ \ket{\alpha_0}, \ket{\beta_0}, \ket{\gamma} \}$ [the $\theta=0$ case of 
\eqref{eqn:MBQC-basis-1}-\eqref{eqn:MBQC-basis-3}] in which the above four $\mathbb{Z}_2 \times \mathbb{Z}_2$ generators are all 
diagonal.    
We set 
\eq{
K_p =
\left(
\begin{array}{ccc}
a_{p} & b_{p} & c_{p}  \\
d_{p} & e_{p} & f_{p} \\
g_{p} & h_{p} & i_{p}
\end{array}
\right)
}
in the same basis, and find $K_{p}$ satisfying the strong symmetry condition \eqref{strong_sym_cond_Kraus} 
for each element of $\mathbb{Z}_2 \times \mathbb{Z}_2$ as:
\begin{itemize}
\item $U_g = \be^{i\pi S^x}$:\\
When $\be^{i \phi(g)}=1$, $b_p = c_p= d_p=g_p=0$ for $^\forall p$, and otherwise $a_p = e_p= f_p=h_p=i_p=0$.
\item $U_h = \be^{i\pi S^y}$:\\
When $\be^{i \phi(g)}=1$, $b_p = d_p= f_p=h_p=0$ for $^\forall p$, and otherwise $a_p = c_p= e_p=g_p=i_p=0$.
\item $U_{gh} = U_{g}U_{h}= \be^{i\pi S^z}$:  \\
When $\be^{i \phi(g)}=1$, $c_p = f_p= g_p=h_p=0$ for $^\forall p$, and otherwise $a_p = b_p= d_p=e_p=i_p=0$.
\end{itemize}
Then, to exclude the possibility of having at least one $\be^{i \phi(g)}= -1$, we consider the following cases:
\begin{description}
\item[Case 1]{
$(\be^{i \phi(g)}, \be^{i \phi(h)}, \be^{i \phi(gh)}) = (1,1,-1)$\\
In this case, only a trivial solution $K_p = O$ for $^\forall p$ is allowed.
}
\item[Case 2]{
$ (\be^{i \phi(g)}, \be^{i \phi(h)}, \be^{i \phi(gh)}) = (1,-1,-1)$\\
In this case, the following non-zero $K_{p}$ are allowed
\eq{
K_p = 
\begin{pmatrix}
0 & 0 & 0  \\
0 & 0 & f_p \\
0 & h_p & 0
\end{pmatrix} 
\quad (\forall p)  \; .}
However, these $K_p$s never satisfy the completeness condition $\sum_p  K^\dagger_p K_p =I_3$ and are not eligible as the Kraus operators.
}
\item[Case 3]{
$( \be^{i \phi(g)}, \be^{i \phi(h)}, \be^{i \phi(gh)} ) = (-1,-1,-1)$\\
In this case, $K_p = O$ for $^\forall p$.
}
\end{description}
We can treat the remaining cases in the same manner to see that when at least one of the phases $\be^{i \phi(g)}$ is $-1$, 
there is no set of the Kraus operators satisfying the strong symmetry condition.  Therefore,  $e^{i \phi(g)} = 1$ for $^\forall g$. 
\end{proof}
From this, the following statements follow immediately. 
\begin{cor}\label{cor1}
If $\{K_p\}$ satisfies the strong symmetry condition for 
the $\mathbb{Z}_2\times \mathbb{Z}_2$ generators $\{ 1, \be^{i\pi \tilde{S}(\theta)^x}, \be^{i\pi \tilde{S}(\theta)^y}, \be^{i\pi \tilde{S}(\theta)^z} \}$, 
then
\begin{itemize}
\item The Kraus operators $\{ K_p \}$ commute with all the above generators $U_{g}$: $U_g K_p  =  K_p U_g$. In particular, 
$\{ K_p \}$ are diagonal in the basis $( \ket{\alpha_\theta}, \ket{\beta_\theta}, \ket{\gamma} )$. 
\item $K_p \mathcal{P}(m) = \mathcal{P}(m) K_p$ for all $p$ and $m$.  
\end{itemize}
\end{cor}
Conversely, if $K_p \mathcal{P}(m) = \mathcal{P}(m) K_p$, both $U_{g}$ and $K_{p}$ are diagonal, which means: $U_g K_p  =  K_p U_g$, 
i.e., $\{K_p\}$ is strongly symmetric for $\{ 1, \be^{i\pi \tilde{S}(\theta)^x}, \be^{i\pi \tilde{S}(\theta)^y}, \be^{i\pi \tilde{S}(\theta)^z} \}$. 
Therefore, $K_p \mathcal{P}(m) = \mathcal{P}(m) K_p$ ~for $^\forall p,m$ and the statement that a quantum channel satisfies the strong symmetry condition for $\mathbb{Z}_2\times \mathbb{Z}_2 = \{ 1, \be^{i\pi S^x}, \be^{i\pi S^y}, \be^{i\pi S^z} \}$ are equivalent. Furthermore, 
\begin{eqnarray}
&&K_p \mathcal{P}(m) = \mathcal{P}(m) K_p \nonumber \\
&&\Longrightarrow \sum_p K^\dagger_p K_p \mathcal{P}(m) = \sum_p K^\dagger_p \mathcal{P}(m) K_p \nonumber \\
&&\Longrightarrow \mathcal{P}(m) = \sum_p K^\dagger_p \mathcal{P}(m) K_p \;.\nonumber 
\end{eqnarray}
On the contrary, if $\sum_p K^\dagger_p \mathcal{P}(m) K_p = \mathcal{P}(m)$,
\begin{eqnarray}
&&\sum_p (P(m)K_p-K_p P(m))^\dagger (P(m)K_p-K_p P(m)) \nonumber \\
&=& \sum_p K_p^\dagger P(m)K_p -P(m) \left( \sum_p K_p^\dagger P(m)K_p \right)  \nonumber \\ 
&&- \left( \sum_p K_p^\dagger P(m)K_p \right)P(m) + P(m)\left( \sum_p K_p^\dagger K_p \right)P(m) \nonumber \\
&=& 0\;,
\end{eqnarray}
and $(P(m)K_p-K_p P(m))^\dagger (P(m)K_p-K_p P(m))$ are positive, so $P(m)K_p=K_p P(m)$. Hence, Lemma\;\ref{lemma} holds.
\rightline{$\blacksquare$}
Finally, the proposition that the gate fidelity of $\be^{-i\theta' Z/2}$ does not decay is equivalent to 
Eq.~\eqref{projection_SS_condition_1} for both $\theta' = 0$ and $\theta'=\theta$.  
So, by Lemma~\ref{lemma}, Prop.~\ref{condition_for_z_rotation} holds.

\section{MPO tensor of decohered SPT state and identity gate}
\label{sec:MPO_tensor_product_decomposition}
In Ref.~\onlinecite{Else-S-B-D-12}, it is shown that ground states in the Haldane phase are uniformly perfect resources of the identity gate by measuring in the common basis $\{\ket{x}:=\ket{\alpha_{\theta=0}},\ket{y}:=i\ket{\beta_{\theta=0}},\ket{z}:=\ket{\gamma}\}$. This is supported by the fact that the MPS tensor $A^{[s]}$ in the same basis has a tensor product decomposition into the logical part $\sigma_s$, which is common throughout the phase and 
encodes quantum information, and the junk part $A_{\text{junk}}^{[s]}$ which contains the microscopic information of the state:
\begin{eqnarray}
A^{[s]} = \sigma_s \otimes A_{\text{junk}}[s]
\end{eqnarray}
 This enables us to implement the identity gate with the Pauli byproduct operator $\sigma_s$ in the logical subspace in the correlation space. 
 This basis $\{\ket{x},\ket{y},\ket{z}\}$ is often called the \textit{wire basis} because the identity gate plays a role of the quantum wire in the quantum circuit.  
 In this section, we extend this fact to mixed states. 

As seen in the last section \ref{sec:proof-Prop2}, if a quantum noise satisfies the strong symmetry condition, its Kraus operators and each element of the linear representation are commutative for the case of $G= \mathbb{Z}_2\times \mathbb{Z}_2$ (Corollary \ref{cor1}). So, the following proposition holds:
\begin{propwonum} 
Let $\ket{\psi}$ be a ground state in the Haldane phase and $\mathcal{E}$ be an uncorrelated noise. The diagonal elements of the MPO tensors 
of $\ketbra{\psi}{\psi}$ and $\mathcal{E}(\ketbra{\psi}{\psi})$ in the wire basis share the same structure 
if and only if $\mathcal{E}$ is strongly symmetric with respect to 
the canonical representation $\{ 1, \be^{i\pi S^x}, \be^{i\pi S^y}, \be^{i\pi S^z} \}$ of the $\mathbb{Z}_2\times \mathbb{Z}_2$-symmetry. 
\end{propwonum}
\begin{proof}
For the clarity of the argument, we now prove the proposition for the case where the initial state $|\psi\rangle$ is the AKLT state.  
As the MPS tensors of the AKLT state in the wire basis are given by the Pauli matrices [see Eq.~\eqref{single_measurement_result}], 
the corresponding MPO tensors are:
\begin{equation}
\raisebox{-3ex}{\includegraphics[scale=0.3,clip]{MPO_tensor_0}} \; 
:= \; 
\raisebox{-4.8ex}{\includegraphics[scale=0.3,clip]{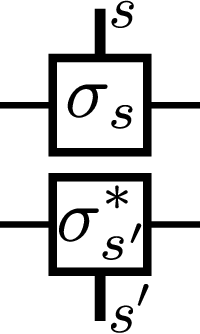}} \; 
= \sigma_s \otimes \sigma^*_s 
 ~~~(s,s' = x, y, z)  \; .
 \label{eqn:MPO-AKLT}
\end{equation}  
Then, the local MPO tensor on each site after an uncorrelated noise $\mathcal{E}=\circ_i \mathcal{E}_i $ is applied reads as:
\begin{equation}
\raisebox{-3ex}{\includegraphics[scale=0.3,clip]{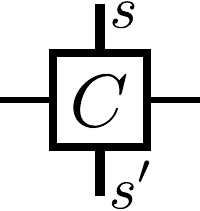}} \; 
:= \; 
\raisebox{-9ex}{\includegraphics[scale=0.3,clip]{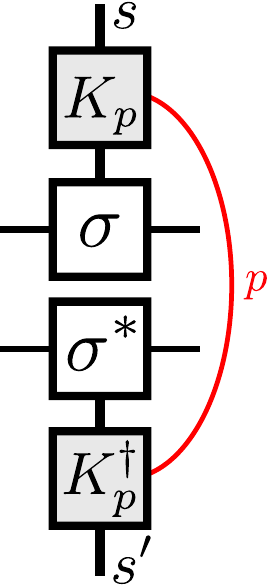}} 
 ~~~(s,s' = x, y, z)  \; ,
\end{equation}
where the $K_p$'s in the left-hand side are the Kraus operators of the local channel $\mathcal{E}_i$.  

We now suppose that the uncorrelated noise $\mathcal{E}$ is strongly symmetric. 
Then, by Corollary \ref{cor1}, these Kraus operators are diagonal in the wire basis 
(in which all the $\mathbb{Z}_2\times \mathbb{Z}_2$ generators are diagonal).   
Therefore, after the projective measurement, the local tensor on the right-hand side changes to (now $s^{\prime}$ is set to $s$ by 
the measurement): 
\begin{equation}
\begin{split}
\raisebox{-2.8ex}{\includegraphics[scale=0.3,clip]{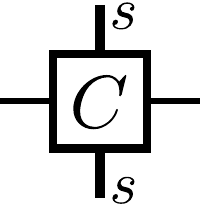}} \; 
&= \sum_p \sum_{t,t^{\prime}} \left( [K_{p}]_{st} \sigma_t \right) \otimes \left( [K^{\dagger}_{p}]_{t^{\prime} s } \sigma_{t^{\prime}}^* \right) \\
&=\sum_p \left( [K_{p}]_{ss} \sigma_s \right) \otimes \left( [K_{p}]^{\ast}_{s s} \sigma_{s }^* \right)   \\
&= (\sum_p  | K_{ss}^{(p)} |^2 ) \sigma_s \otimes \sigma^*_s   
= \sigma_s \otimes \sigma^*_s~~~(s=x,y,z)  \;  .
\end{split}
\end{equation}
Here, the last equality arises from $\sum_p K^\dagger_p K_p =I$.
Hence, on the AKLT state decohered by a strongly symmetric noise, the projective measurement in the wire basis realizes the identity gate 
with the same byproduct operator $\sigma_{s}$ as the pure case. 

Conversely, if the local MPO tensor assumes the following simple form:
\eq{
\raisebox{-3.3ex}{\includegraphics[scale=0.35,clip]{MPO_tensor_3}} 
\; = \;
\sigma_s \otimes \sigma^*_s  \; 
}
for all $s$, then 
\begin{equation}
\sum_p \left(\sum_l [K_p]_{sl}\sigma_l \right) \otimes \left(\sum_{l'}\sigma^*_{l'} [K^\dagger_p]_{l's}\right)
= \sigma_s \otimes \sigma^*_s 
\end{equation}
must hold for all $s$.  
By the linear independence of the basis $\sigma_{s} \otimes \sigma_{s^{\prime}}$, 
the above implies $[ K_p ]_{sl} =0~(l \neq s)$ for all $p$.   
Therefore, the Kraus operators are all diagonal in the wire basis, and, by Corollary~\ref{cor1}, this quantum channel is strongly symmetric.
This proof can be extended straightforwardly to generic ground states in the Haldane phase by replacing the MPS tensor $A^{[s]}=\sigma_s$ with  
$A^{[s]} = \sigma_s \otimes A_{\text{junk}}[s]$; again the diagonal element is invariant under the quantum channel if and only 
if it is strongly symmetric. 
\end{proof}

As the proposition implies that the logical space, which is seen in the form of the MPO tensor \eqref{eqn:MPO-AKLT}, 
is preserved by the quantum channel, the two states $\ketbra{\psi}{\psi}$ and $\mathcal{E}(\ketbra{\psi}{\psi})$ 
share the same computational power for the identity gate.   
Thus we arrive at the statement given in Prop.~\ref{MPOtensor_and_computaitonal_power}:
\begin{cor} \label{identity_gate_and_MPO}
The measurement in the wire basis implements the identity gate with fidelity $1$ on the mixed states 
that result from any ground states in the Haldane phase if and only if the applied quantum channel satisfies strong symmetry condition 
for $\mathbb{Z}_2\times \mathbb{Z}_2=\{ 1, \be^{i\pi S^x}, \be^{i\pi S^y}, \be^{i\pi S^z} \}$.
\end{cor}

%
\end{document}